\documentclass[journal]{IEEEtran}
\usepackage{graphics}
\usepackage{graphicx}
\usepackage{amsfonts}
\usepackage{multirow}
\usepackage{epstopdf}
\usepackage{color}
\usepackage{float}
\usepackage[font=small]{caption}
\usepackage{bm}
\usepackage{algorithm}
\usepackage{algorithmic}
\usepackage{amsmath}
\usepackage{pifont}
\usepackage[utf8]{inputenc}
\usepackage{mathtools}
\usepackage{amssymb}
\usepackage{amsthm}

\usepackage[table]{xcolor}
\newcolumntype{s}{>{\columncolor[HTML]{FE6F5E}} c}
\newcolumntype{t}{>{\columncolor[HTML]{5D8AA8}} c}

\usepackage{booktabs}

\newtheorem{theorem}{Theorem}
\newtheorem{corollary}{Corollary}
\newtheorem{lemma}{Lemma}

\ifCLASSINFOpdf
\else
\fi

\begin{document}

\title{Sparse Doppler Sensing Based on Nested Arrays}

\author{Regev~Cohen,~\IEEEmembership{Student,~IEEE,}
        Yonina~C.~Eldar,~\IEEEmembership{Fellow,~IEEE}
\thanks{This work was funded by the European Union’s Horizon 2020 research and innovation
program under grant agreement No. 646804-ERC-COG-BNYQ.}}


\maketitle

\begin{abstract}
Spectral Doppler ultrasound imaging allows visualizing blood
flow by estimating its velocity distribution over time. Duplex ultrasound is a modality in which an ultrasound system is used for displaying simultaneously both B-mode images and spectral Doppler data. In B-mode imaging short wide-band pulses are used to achieve sufficient spatial resolution in the images. In contrast, for Doppler imaging, narrow-band pulses are preferred in order to attain increased spectral resolution. Thus, the acquisition time must be shared between the two sequences. In this work, we propose a non-uniform slow-time transmission scheme for spectral Doppler, based on nested arrays, which reduces the number of pulses needed for accurate spectrum recovery. We derive the minimal number of Doppler emissions needed, using this approach, for perfect reconstruction of the blood spectrum in a noise-free environment. Next, we provide two spectrum recovery techniques which achieve this minimal number. The first method performs efficient recovery based on the fast Fourier transform. The second allows for continuous recovery of the Doppler frequencies, thus avoiding off-grid error leakage, at the expense of increased complexity. The performance of the techniques is evaluated using realistic Field II simulations as well as in vivo measurements, producing accurate spectrograms of the blood velocities using a significant reduced number of transmissions. The time gained, where no Doppler pulses are sent, can be used to enable the display of both blood velocities and high quality B-mode images at a high frame rate.  
\end{abstract}


\IEEEpeerreviewmaketitle

\section{Introduction}

\IEEEPARstart{S}{pectral} Doppler in medical ultrasound is a non-invasive imaging modality commonly used for quantitative estimation of blood velocity. The data for velocity estimation is acquired by insonifying the medium with a train of narrow-band ultrasound pulses along a desired direction at a constant pulse repetition frequency (PRF). The backscattered signals are then sampled and focused along the chosen direction using dynamic focusing. Assembling the samples associated with a specific depth of interest from all received signals forms the so-called slow-time signal with a center frequency proportional to the axial blood velocity. 

For a single blood cell with axial velocity $v_{z}$, the slow-time signal has a center frequency equal to \cite{jensen1996estimation}
\begin{equation}
f_{\text{D}}=-\frac{2v_{z}}{c}f_{0}
\end{equation}
where $f_{0}$ is the center frequency of the transmitted signal and $c$ is the speed of sound.
In reality, there is a distribution of blood scatterers within each resolution cell of the ultrasound system. The blood velocity distribution is estimated by reconstructing the power spectral density (PSD) of the slow-time signal. Displaying spectral analysis results over time on a pulsed Doppler spectrogram (also referred to as pulsed wave spectrogram), visualizes the evolution of the blood velocity distribution as a function of time. The time needed for each velocity estimation is the coherent processing interval (CPI), which is equal to the number of transmitted pulses $P$ divided by the PRF. As the number of transmitted pulses per unit time is limited by the speed of sound and the desired depth being examined, there is an inherent trade off between spectral and temporal resolution.

In modern commercial ultrasound systems, the spectrogram is typically estimated using Welch’s method \cite{welch1967use,stoica2005spectral}, a modified averaged periodogram based on the fast Fourier transform (FFT). However, this approach suffers from high leakage due to high sidelobes and/or low resolution. Since the resolution in Doppler frequency is governed by $P$, it requires a large number of consecutive transmissions to be used for each velocity estimate. 

In addition to Doppler measurements, simultaneous high frame rate B-mode images are required to allow the physician to navigate, select the region in which the blood velocity is estimated and to examine anatomical structures surrounding the vessel. However,  two distinct pulses are used for the two modes, B-mode and Doppler. In particular, for B-mode imaging short wide-band pulses with high carrier frequency are transmitted to increase resolution. Whereas, for Doppler imaging, narrow-band pulses with low center frequency are preferred in order to improve penetration depth and increase the precision of the velocity estimate. Moreover, the B-mode and Doppler pulses may be transmitted in different directions. Consequently, the acquisition time must be shared between the two imaging modalities.

In conventional imaging, an interleaved B-mode/Doppler sequence is used where every B-mode transmission is followed by a Doppler transmission. This halves the PRF, resulting in reduction of the maximal velocity that can be detected by a factor of two, according to the Nyquist theorem.  An alternative common approach is to regularly interrupt the Doppler sequence for a block of B-mode transmissions. However, this results in holes in the blood velocity spectrogram. These limitations raise the need for developing improved techniques for blood spectrum estimation using considerably fewer Doppler transmissions.

To circumvent these problems Kristoffersen and Angelsen \cite{kristoffersen1988time} proposed to fill in the Doppler gaps with a synthetic signal, generated based on the Doppler signal measured immediately prior to the B-mode interrupt.  Klebaek et al. \cite{klebaek1995neural} proposed the use of neural networks for predicting the evolution of the mean and variance of the Doppler signal in the gaps. However, both methods are based on the assumption that the blood flow is constant or predictable which is not true in case of abrupt changes, leading to inaccurate velocity estimation. A correlation-based method for spectral estimation from sparse data sets was proposed in \cite{jensen2006spectral}, allowing for random Doppler transmission schemes, but it requires long ensembles to avoid aliasing. This work was further investigated in \cite{mollenbach2008duplex}, which proposed proposing a technique for reconstructing the missing Doppler samples, due to B-mode transmissions, using filter banks. This method, however, reduces the velocity range in proportion to the number of missing Doppler samples.  

Two data-adaptive velocity estimators for periodically gapped data, called BPG-Capon and BPG-APES, were suggested in \cite{liu2009periodically,larsson2003spectral}. These methods are restricted to the case of periodically gapped sampling of Doppler emissions and have been shown to achieve a limited reduction of 34\% in the number of transmissions. For arbitrary Doppler subsampling patterns, two iterative methods termed BSLIM and BIAA were presented in \cite{gudmundson2011blood,gran2009adaptive}. However, they exhibit high computational load and require the use of regression filters for clutter removal, which may degrade the quality of the spectrum estimate by producing spurious frequency components \cite{torp1997clutter}. 

Several works apply compress sensing (CS) \cite{eldar2012compressed} techniques to spectral Doppler using random slow-time samples.  Zobly et al. apply basis pursuit (BP) in \cite{zobly2011compressed} and a multiple measurement vector (MMV) technique in \cite{zobly2013multiple} to recover the Doppler signal. However, the authors do not state the domain (dictionary) in which the signal is sparse. Furthermore, the resultant spectrograms exhibit artifacts. Assuming the Doppler signal is sparse under the Fourier transform or in the wave atom domain \cite{demanet2007wave}, Richy et al. propose \cite{richy2013blood,richy2011blood} decomposing the Doppler signal into several equal segments and applying CS recovery on each segment. However, this work does not consider the case of moderately or non sparse signals. Moreover, the reduction in the number of Doppler transmissions is limited to 60\% using this method. An extension of this study is presented in \cite{lorintiu2016compressed}, which proposed to reconstruct the Doppler signal using block sparse Bayesian learning (BSBL) \cite{zhang2013extension,zhang2012recovery}. However, the authors assume that the Doppler samples are temporally correlated and severe aliasing appears in their recovered spectra at high subsampling rates. In addition, the average computation time per segment using this technique is high, making it impractical for real-time implementation. 

In addition to the computational complexity and recovery artifacts in the methods above, none of these works present an analysis of the minimal number of Doppler emissions ensuring adequate reconstruction of the blood spectrum, using their techniques. 

The main contribution of this paper is twofold. First, adopting recent work on nested arrays \cite{pal2010nested,pal2010novel} in the fields of multiple-input multiple-output (MIMO) radar systems and direction of arrival (DOA) estimation, we present a non-uniform transmission scheme for spectral Doppler. Our theoretical approach does not assume the Doppler signal is sparse or its entries are correlated, nor that the blood flow is predictable. An analysis is performed, deriving the minimal number of Doppler emissions required using the nested approach. We show that the number of transmissions allowing for perfect reconstruction of the spectrum in a noise-free setting is proportional to the square root of the observation window length. Second, we propose two spectrum recovery techniques that achieve this minimal number of transmitted pulses. The first method assumes the Doppler frequencies lie on the Nyquist grid and recovers the spectrum using FFT. This technique exhibits enhanced resolution compared with Welch's method, and similar low complexity, making it suitable for real-time application. The second approach performs continuous recovery, thus preventing spectral leakage stemming from off-grid errors, at the expense of increased complexity. The performance of the techniques is validated using realistic Field II simulation data \cite{jensen1996field,jensen1992calculation} and in vivo data, showing that blood velocities can be accurately estimated from a reduced number of emissions.

The rest of the paper is organized as follows. In Section \ref{sec:model}, we review the Doppler signal model and formulate our problem. Section\,\ref{sec:sampling} describes the autocorrelation of the Doppler signal and introduces the proposed sparse slow-time sampling scheme. We then derive the minimal number of Doppler transmissions required using this emission pattern.  In Section \ref{sec:recon}, we present discrete and continuous recovery techniques that achieve this minimal number. Alternative sparse transmission schemes are discussed in Section \ref{sec:alternatives}. We evaluate the performance of the proposed algorithms in Section \ref{sec:sim} and compare them with existing state-of-the-art techniques. Finally, Section \ref{sec:con} concludes the paper.  

Throughput the paper we use the following notation. Scalars  are denoted by lowercase letters $(a)$, vectors by boldface lowercase letters $(\bf a)$, matrices by boldface capital letters $(\bf A)$ and sets are given by calligraphic font  (e.g., $\mathcal{A}$). The $(i,j)$th element of $\bf A$ is denoted by ${\bf A}(i,j)$, ${\bf a}_l$ is the $l$th column of ${\bf A}$ and ${\bf a}(l)$ represents the $l$th element of ${\bf a}$. The notations $(\cdot)^T,(\cdot)^\ast$ and $(\cdot)^H$  indicate the transpose, conjugate and Hermitian operations, respectively. The vectorization of  a matrix $\bf A$ into a column stack is given by $\text{vec}({\bf A})$. For a positive integer $P$, $d|P$ implies that $d$ is a divisor of $P$ with  $1<d<P$.

\section{Doppler Model and Problem Formulation  }
\label{sec:model}
\subsection{Doppler Model}
A standard ultrasound system in spectral Doppler mode transmits a pulse train
\begin{equation}
s_{tx}(t)=\sum_{p=0}^{P-1}h(t-pT),\qquad 0\leq t\leq PT,
\label{eq:tx}
\end{equation} 
consisting of $P$ equally spaced pulses{} $h(t)$. The pulse repetition interval (PRI) is $T$, and
its reciprocal $f_{\text{prf}}=1/T$ is the PRF. The entire span of the signal in (\ref{eq:tx}) is called the CPI. The pulse $h(t)$ is a sinusoid defined as
\begin{equation}
h(t)=\sin(2\pi f_{0}t),\quad 0\leq t\leq T_{\max},
\end{equation}
where $f_{0}$ is the center frequency of the signal and $T_{\max} < T$ is the pulse duration, determined by the maximal depth examined.

Consider a single blood scatterer. The pulses reflect off the scatterer and propagate back to the transducer. The noise-free received signal can be modeled as
\begin{equation}
s(t)=\sum_{p=0}^{P-1}\alpha \sin \bigg(2\pi f_{0}\Big(t-pT-\frac{2d_{p}}{c}\Big)\bigg),
\label{eq:stx}
\end{equation} 
where $p$ is the emission number, $c$ is the sound wave propagation speed, $\alpha$ is the amplitude related to blood scatterer reflectivity and $d_{p}$ is its depth at the time of the $p$th transmission. For mathematical convenience, we express $s(t)$ as a sum of single frames
\begin{equation}
s(t)=\sum_{p=0}^{P-1}s_{p}(t),
\end{equation}
where
\begin{equation}
s_{p}(t)=\alpha \sin \bigg(2\pi f_{0}\Big(t-pT-\frac{2d_{p}}{c}\Big)\bigg).
\label{eq:sp}
\end{equation}

The blood scatterer movement along the beam direction during $P$ consecutive transmissions is given by
\begin{equation}
d_{p}=d_{0}+v\cdot pT,\qquad 0\leq p\leq P-1,
\label{eq:dp}
\end{equation}
where $d_{0}$ is the initial depth of the blood scatterer and $v$ is its axial velocity. Substituting (\ref{eq:dp}) into (\ref{eq:sp}), we get
\begin{equation}
s_{p}(t)=\alpha \sin \bigg(2\pi f_{0}\Big(t-pT-\frac{2d_{0}}{c}-\frac{2v}{c}pT\Big)\bigg).
\end{equation}
Each frame is then aligned $\tilde{s}_{p}(t)=s_{p}(t+pT)$ and sampled at rate $f_s$, determined by the desired spatial axial resolution. This yields a 2D discrete signal
\begin{equation}
s[k,p]=\tilde{s}_{p}\Big(\frac{k}{f_{s}}\Big)=\alpha \sin \bigg(2\pi f_{0}\Big(\frac{k}{f_{s}}-\frac{2d_{0}}{c}-\frac{2v}{c}pT\Big)\bigg),
\label{eq:2dsamples}
\end{equation}
where $k$ is the sample index associated with depth.

The samples (\ref{eq:2dsamples}) form a 2D measurement matrix $\mathbf{S}\in \mathbb{C}^{K\times P}$  where $\mathbf{S}(k,p)=s[k,p]$. For a fixed pulse number $p$, the samples along the row dimension of $\mathbf{S}$ are referred to as fast-time samples and are related to the $p$th pulse transmission. Each fast-time sample corresponds to a different depth $k$ of the scanned medium. For a given $k$, the samples along the column dimension of $\mathbf{S}$ are referred to as slow-time samples and are associated with the same depth, one sample per pulse emission.

Following (\ref{eq:2dsamples}), the analytical signal is generated to give the in-phase and quadrature components
\begin{align}
\begin{split}
x[k,p] &=s[k,p]+j\mathcal{H}_{k}\{s[k,p]\}= \\
&=\alpha \exp \bigg(2\pi jf_{0}\Big(\frac{k}{f_{s}}-\frac{2d_{0}}{c}-\frac{2v}{c}pT\Big)\bigg),
\end{split}
\end{align}
where $\mathcal{H}_{k}\{\cdot\}$ is the discrete Hilbert transform in the fast-time direction. Since $f_{0}/f_{s}$ is known, we demodulate the signal $x[k,p]$, resulting in
\begin{equation}
y[k,p]=\alpha \exp \bigg(-2\pi jf_{0}\Big(\frac{2d_{0}}{c}+\frac{2v}{c}pT\Big)\bigg).
\label{eq:y_kp}
\end{equation}
Define the complex amplitude $\tilde{\alpha}=\alpha \exp (-j\frac{4\pi d_{0}}{c}f_{0})$ and
denote the Doppler frequency by 
\begin{equation}
f\triangleq -\frac{2v}{c}f_{0}.
\end{equation}
Then, we can represent the signal given in (\ref{eq:y_kp}) as
\begin{equation}
y[k,p]=\tilde{\alpha} \exp \big(2\pi jf pT\big).
\label{eq:y_kp2}
\end{equation}

Consider  a specific depth $k$. The measured signal in (\ref{eq:y_kp2}) can be viewed as a realization of a continuous-time wide-sense stationary (WSS), comprises a zero-mean complex amplitude amplitude and a time invariant velocity 
\begin{equation}
y_{k}(t)=\tilde{\alpha} \exp (2\pi jf t),
\end{equation}
which is sampled at time $t=pT$ $(0\leq p\leq P-1)$, namely, at a sampling rate of $f_{\text{prf}}$.  Decreasing the sampling interval $T$ increases the maximal velocity that can be recovered according to the Nyquist theorem, however, there is a trade-off since it limits the maximal depth being examined. In addition, the spectral resolution is governed by $P$, motivating the desire to increase the number of transmissions as long as they are limited to be within the time when the velocity is assumed to be constant.  

In the general case, each resolution cell of the ultrasound imaging system contains a distribution of  blood scatterers. Consequently, the measured signal consists of $M>1$ unknown frequencies $\{f_{m}\}_{m=1}^{M}$. Taking the latter into account, we extend the signal model written in (\ref{eq:y_kp2}) to
\begin{equation}
y[k,p]=\sum_{m=1}^{M} \alpha_{m}\exp(2\pi jf_{m}pT),\qquad 0\leq p\leq P-1.
\end{equation}
Therefore, the received signal is composed of $M$ components where the $m$th component is defined by two parameters: a Doppler frequency $f_{m}$, proportional to an axial velocity $v_{m}$; and a complex random amplitude $\alpha_{m}$, related to the number of blood cells moving at an axial velocity $v_{m}$ and their positions. The Doppler frequencies $\{f_{m}\}_{m=1}^{M}$ are assumed to lie in the unambiguous frequency domain, that is $|f_{m}|\leq \frac{1}{2T}=\frac{1}{2}f_{\text{prf}}$ for all $1\leq m\leq M$.

Assembling the slow-time samples $y[k,p]$ for $P$ consecutive transmissions into a vector we obtain
\begin{equation}
{\bf y}[k]=\bf{A}\boldsymbol{\alpha},
\label{eq:y=Ax}
\end{equation}
where ${\bf y}[k]=[y[k,0],y[k,1],...,y[k,P-1]\,]^T\in\mathbb{C}^{P\times 1}$ is the slow-time vector, the vector $\boldsymbol{\alpha}\in\mathbb{C}^{M\times 1}$ consists of $M$ amplitudes $\{\alpha_m\}_{m=1}^M$ and the matrix ${\bf A}\in\mathbb{C}^{P\times M}$ is a Vandermonde matrix, whose entries are given by ${\bf A}(p,m)=\exp(2\pi jf_mpT)$. 

Based on the model (\ref{eq:y=Ax}), the goal is to recover the frequency components $\{f_{m}\}_{m=1}^{M}$  which form the matrix  $ \bf A $ and to estimate the variances $\{\sigma_{m}^2\}_{m=1}^{M}$ of the random vector $\boldsymbol{\alpha}$, i.e.,  the power spectrum.   

\subsection{Standard Processing}

In standard Doppler processing \cite{jensen1996estimation,welch1967use}, the Doppler frequencies are assumed to lie on the Nyquist grid, that is $f_mT=i_m/P$, where $i_m$ is an integer in the range $0\leq i_m\leq P-1$. Using this assumption, (\ref{eq:y=Ax}) can be rewritten with ${\bf A}={\bf F}^H$ as
 \begin{equation}
 {\bf y}[k]={\bf F}^H\boldsymbol{\alpha},
 \end{equation}
  where ${\bf F}\in\mathbb{C}^{P\times P}$ is the FFT matrix. This implies that $\boldsymbol{\alpha}$ is a vector of length $P$ with $M$ non-zero values $\{\alpha_m\}_{m=1}^M$ at indices  $\{i_m\}_{m=1}^M$. Consequently, the power spectrum, to be recovered, is defined as a vector  $\bf p\in\mathbb{R}^{P\times 1}$ with a non-zero value $\sigma_m^2$ at index $i_m$.

  Assuming we have enough snapshots of the slow-time vector ${\bf y}[k]$, a conventional estimate of the power spectrum is given by
  \begin{equation}
  \hat{\bf p}_\text{standard}=\frac{1}{K}\sum_{k=1}^K |{\bf F}{\bf y}[k]|^2,
  \end{equation}
  where the squared magnitude is computed element-wise.
 In this case, the spectral resolution is equal to $2\pi/PT$, where $P$ is chosen large enough to attain sufficient resolution.

\subsection{Problem Formulation}
In this work, we wish to recover the power spectrum $\bf p$ with improved spectral resolution while significantly reducing the number of transmitted Doppler pulses.

For an observation window of size $P$, we propose a  new transmission strategy  in which only $N<P$ pulses are sent with non-uniform time steps between them over the entire CPI. We show that the power spectrum can be fully reconstructed with a resolution of $2\pi/(2P-1)T$ at the same complexity of standard processing. Note that we do not recover the slow-time signal but only its power spectrum. We prove that $N = 2 \sqrt{P}-1$ is the minimal number of transmissions enabling perfect reconstruction of the spectrum in a noise-free environment using our approach, and present recovery techniques that achieve this number.  

Using our techniques, we allow periods of time where no Doppler pulse is sent, which can be exploited  for B-mode transmission sequences.
Consequently, the same CPI may be used to achieve Doppler velocity estimates and high quality B-mode images at a high frame rate. 

\section{Nested Slow-Time Sampling}
\label{sec:sampling}

In this section, we present a non-uniform Doppler transmission scheme from which the blood spectrum may be recovered with improved resolution, in comparison to standard processing. We first extend the signal model (\ref{eq:y=Ax}) by deriving an expression for the signal autocorrelation function. 

\subsection{Correlation Domain}
\label{subsec:corrdom}
Consider the model given by (\ref{eq:y=Ax}) and define the autocorrelation matrices ${\bf R_y}=\mathbb{E}[{\bf yy}^H]\in\mathbb{C}^{P\times P}$ and ${\bf R}_{\boldsymbol{\alpha}}=\mathbb{E}[\boldsymbol{\alpha\alpha}^H]\in\mathbb{C}^{M\times M}$. Then,
\begin{equation}
{\bf R_y}= {\bf AR_{\boldsymbol{\alpha}}A}^H.
\label{eq:R=ARA}
\end{equation}
We  further assume that the amplitudes are statistically uncorrelated with unknown variances such that
\begin{equation}
\mathbb{E}[\alpha_{m}\alpha_{n}^{\ast}]=\sigma_{m}^{2}\delta[n-m],
\end{equation} 
where $\delta[\cdot]$ is the Kronecker delta. Under this assumption, the matrix ${\bf R}_{\boldsymbol{\alpha}}$ is a diagonal matrix with ${\bf R}_{\boldsymbol{\alpha}}(m,m)=\sigma_m^2$. Denoting the diagonal of ${\bf R}_{\boldsymbol{\alpha}}$ by $\bf p\in\mathbb{R}^{M\times 1}$, it follows that
\begin{equation}
{\bf r}\triangleq\text{vec}({\bf R_y})=({\bf A}^\ast \odot{\bf A}){\bf p},
\label{eq:r=Ap}
\end{equation}
where ${\bf A}^\ast \odot{\bf A}\in\mathbb{C}^{P^2\times M}$ and  $\odot$ denotes the Khatri-Rao product defined as a column-wise Kronecker product between two matrices with the same number of columns \cite{van2002optimum,ma2009doa}.

For a Vandermonde matrix $\bf A$ defined as in (\ref{eq:y=Ax}), the matrix ${\bf A}^\ast \odot{\bf A}$ has full column rank if  $M\leq2P-1$ \cite{ma2010doa}. Therefore, assuming this condition holds, (\ref{eq:r=Ap}) can be solved uniquely, i.e., we can recover the blood spectrum $\bf p$. Moreover, this condition allows to recover $\bf p$ while transmitting fewer Doppler pulses, as we show in the next subsection. 

\subsection{Nested Transmission Scheme}
\label{subsec:sampling}

We now present a Doppler transmission scheme based on the concept of nested arrays \cite{pal2010nested,pal2012nested,pal2012nested-b}, which has recently been considered in the fields of MIMO radar and DOA. A nested array is an array geometry obtained by systematically nesting two uniform linear arrays (ULA), which allows to resolve $O(N^2)$ signal sources using only $N$ physical sensors when the second-order statistics of the received data is used. We adopt this concept and modify it for Doppler emissions with a fixed CPI (i.e., limited aperture).  

Following the work in  \cite{pal2010nested}, we introduce two positive integers $N_1,N_2$ in the range $1\leq N_1,N_2\leq P$ such that
\begin{equation}
N_2(N_1+1)=P.
\label{eq:N2(N1+1)=P}
\end{equation}
 We then choose the number of pulses to be $N=N_1+N_2$. Notice that $N=N_1+N_2\leq P$  for any two positive integers satisfying (\ref{eq:N2(N1+1)=P}). In the next section we will show how to choose $N_1,N_2$ in order to minimize $N$. 

 Given $N_1$ and $N_2$, we define the following two sets
 \begin{align}
\begin{split}
\mathcal{S}_{\text{N}_1}&=\{1,2,...,N_1\}, \\
\mathcal{S}_{\text{N}_2}&=\{n(N_1+1), \quad\,  n=1,2,...,N_2\}.
\end{split}
 \end{align}
Denote by $\mathcal{S}_\text{N}$  the ordered set of the union of $\mathcal{S}_{\text{N}_1}$ and $\mathcal{S}_{\text{N}_2}$ 
\begin{equation}
\mathcal{S}_\text{N}=\{\mathcal{S}_{\text{N}_1}\cup\mathcal{S}_{\text{N}_2} \},
\end{equation}
which is referred to as a nested array.
By varying $N_1$ and $N_2$ we generate different sets $\mathcal{S_\text{N}}$. Any set in this class is a concatenation of two ULAs with increasing inter-element spacing. Note that for $N_1=P-1$ and $N_2=1$ we have $\mathcal{S}_\text{N}=\{1,2,...,P\}$, hence, the standard transmission pattern is a special case of nested arrays.

Consider a non-uniform transmission pattern for spectral Doppler imaging such that the $n$th pulse is sent at time $p_nT$, where $p_n$ is the $n$th element of $\mathcal{S}_\text{N}$, as illustrated in Fig.\,\ref{fig:patterns}. 
\begin{figure}[h]
 \centering
 \includegraphics[clip,width=1\linewidth]{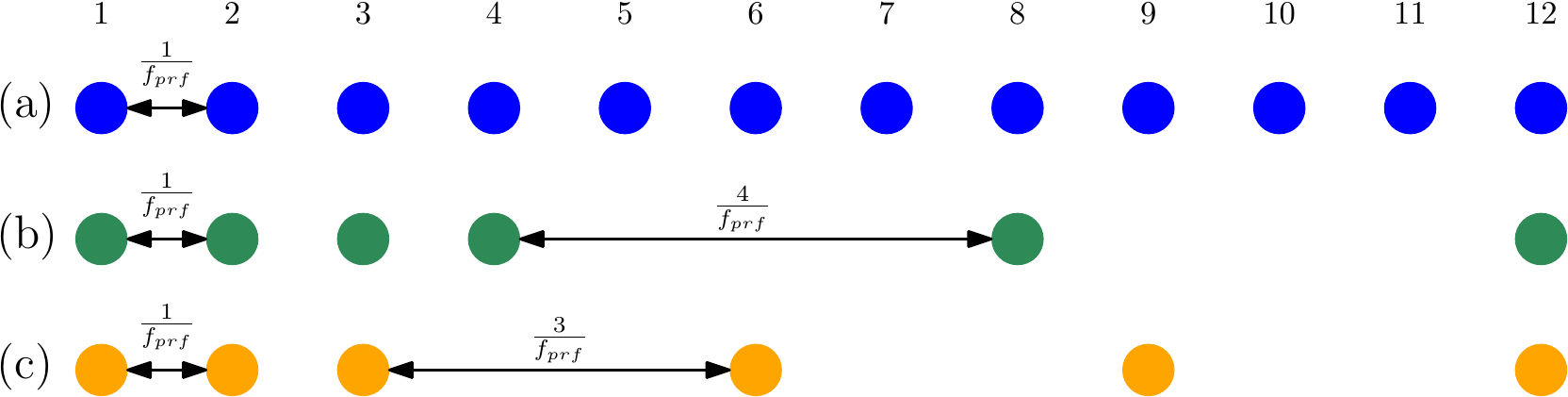}
 \caption{{\bf Transmission Patterns}. Different transmission patterns for an observation window of size $P=12$, where every circle represents a Doppler pulse emission. (a) Standard transmission pattern. (b) Nested transmission pattern for $N_1=N_2=3$. (c) Nested transmission pattern for $N_1=2$ and $N_2=4$.}
  \label{fig:patterns}
 \end{figure}
In this case, (\ref{eq:stx}) becomes
\begin{equation}
s_{tx}(t)=\sum_{n=0}^{N-1}\sin \Big(2\pi f_{0}\big(t-(p_n-1)T\big)\Big),\qquad 0\leq t\leq PT.
\end{equation}
Following the processing on the received signals described in Section \ref{sec:model}, the measured signal is written similarly to (16) as
\begin{equation}
y[k,n]=\sum_{m=1}^{M} \alpha_{m}\exp\big(2\pi jf_{m}(p_n-1)T\big),\quad 0\leq n\leq N-1.
\end{equation}
In vector form we have
\begin{equation}
{\bf y}_\text{N}[k]={\bf A}_\text{N}\boldsymbol{\alpha},
\label{eq:main}
\end{equation}
where ${\bf y}_\text{N}[k]\in\mathbb{C}^{N\times 1}$ is the nested slow-time vector composed of samples from $N$ emissions and ${\bf A}_\text{N}\in\mathbb{C}^{N\times M}$ is a matrix whose entries are given by ${\bf A}(n,m)=\exp\big(2\pi jf_m(p_n-1)T\big)$. Note that ${\bf A}_\text{N}$ is constructed by choosing rows from the Vandermonde matrix $\bf A$, defined in (\ref{eq:y=Ax}), according to $\mathcal{S}_\text{N}$. 

Denote the autocorrelation matrix ${\bf R}_{{\bf y}_\text{N}}=\mathbb{E}[{\bf y}_\text{N}{\bf y}_\text{N}^H]\in\mathbb{R}^{N\times N}$. Similarly to (\ref{eq:R=ARA}) and (\ref{eq:r=Ap}), we have
 \begin{gather}
 {\bf R}_{{\bf y}_\text{N}}= {\bf A_\text{N}R_{\boldsymbol{\alpha}}A}_\text{N}^H, \label{eq:covmtx}\\
{\bf r}_\text{N}\triangleq\text{vec}({\bf R}_{{\bf y}_\text{N}})=({\bf A}_\text{N}^\ast \odot{\bf A}_\text{N}){\bf p}\triangleq\tilde{\bf A}{\bf p}. 
\label{eq:khtrao}
 \end{gather}
In (\ref{eq:khtrao}), the $m$th column of the matrix $\tilde{\bf A}\in\mathbb{C}^{N^2\times M}$ has entries $\exp\big(2\pi jf_m(p_1-p_2)\big)$ for $p_1,p_2\in\mathcal{S}_\text{N}$, where $p_1$ and $p_2$ are pulse locations in the nested array $\mathcal{S}_\text{N}$. Defining the difference set of $\mathcal{S}_\text{N}$ as
 \begin{equation}
 \mathcal{D}=\{p_i-p_j|\quad p_i,p_j\in\mathcal{S}_\text{N}\},
 \end{equation}
 the entries of $\tilde{\bf A}$ are given by $\tilde{\bf A}(d,m)=\exp(2\pi jf_mp_dT)$ where $p_d$ is the $d$th element of $\mathcal{D}$. Note that in our definition of $\mathcal{D}$, we allow repetition of its elements. 

 The system of equations defined in (\ref{eq:khtrao}) can be solved uniquely if the matrix $\tilde{\bf A}$ has  full column rank. Theorem\,1 states necessary  conditions for unique recovery. The theorem relies on the following lemma.

 \begin{lemma}
 \label{lem:diffset}
 Let $\mathcal{D}_u$ be the set of unique elements of $\mathcal{D}$. Then, $\mathcal{D}_u$ consists of exactly $2N_2(N_1+1)-1$ distinct integers in the continuous range  from $-N_2(N_1+1)+1$ to $N_2(N_1+1)-1$. 
 \end{lemma}

 \begin{proof}
See Appendix \ref{app:lemma1}. 
 \end{proof}
 The number of degrees of freedom (DOF) of  the nested set $\mathcal{S}_\text{N}$ is defined as the cardinality of the set $\mathcal{D}_u$.  In our case, according to Lemma \ref{lem:diffset}, the cardinality is equal to $|\mathcal{D}_u|=2N_2(N_1+1)-1$. This number dictates the DOF of the system defined in (\ref{eq:khtrao}) as stated in the next theorem, which follows directly from Lemma 1.
 \begin{theorem}
 \label{theo:fullrank}
 Let ${\bf A}_\text{N}\in\mathbb{C}^{N\times M}$  be the matrix defined in (\ref{eq:covmtx}) with $|f_m|\leq\frac{1}{2}f_\text{prf}$, $1\leq m\leq M$. Then, the matrix $\tilde{\bf A}\triangleq({\bf A}_\text{N}^\ast \odot{\bf A}_\text{N})\in\mathbb{C}^{N^2\times M}$ has exactly $2P-1$ distinct rows. It has full column rank if $2P>M$.
 \end{theorem}

 \begin{proof}
Recall that the entries of $\tilde{\bf A}$ are given by  $\tilde{\bf A}(d,m)=\exp(2\pi jf_mp_d)$. This implies that the $d$th row of $\tilde{\bf A}$ corresponds to the $d$th element of the difference set $\mathcal{D}$. Consequently, the number of distinct rows of $\tilde{\bf A}$ is  equivalent  to the number of unique elements of $\mathcal{D}$, which from Lemma \ref{lem:diffset} is $2N_2(N_1+1)-1$. Since $N_2(N_1+1)=P$, the matrix $\tilde{\bf A}$ has $2P-1$ distinct rows which correspond to a Vandermonde matrix. Hence,  $\tilde{\bf A}$ is full column rank if $2P>M$.  
 \end{proof}

 We can relate each element of the set $\mathcal{D}_u$ to a different time lag of the autocorrelation function of the slow-time signal. Thus, Lemma \ref{lem:diffset}, followed by Theorem \ref{theo:fullrank}, ensures the recovery of all time lags of the autocorrelation function. This means that for $2P>M$, we can retrieve the power spectrum of the slow-time signal by exploiting its stationarity property and the lack of correlation between the amplitudes. As we probe below in Theorem \ref{theo:mintrans}, this may occur ever for $N<P$.  

\subsection{Minimal Sampling Rate}
We next derive the minimal number of Doppler transmissions which allow perfect spectrum recovery while using the nested emission scheme introduced in Subsection $\ref{subsec:sampling}$.

Given an observation window of size $P$, we seek integers $N_1$ and $N_2$ which minimize the total number of Doppler transmissions $N$ while maintaining the overall CPI. This can be cast as the following optimization problem:
\begin{align}
\begin{split}
&\underset{N_1,N_2\in\mathbb{N}^{+}}{\min}\qquad N_1+N_2 \\ 
&\text{subject to}\quad N_2(N_1+1)=P.
\end{split}
\label{eq:mintrans}
\end{align} 
Note that whenever $P$ is a prime number there is only one feasible solution, and hence it is optimal, which is $N_1=P-1$ and $N_2=1$, leading to the standard transmission scheme. Therefore,  we treat the case in which $P$ is not prime and (\ref{eq:mintrans}) becomes a combinatorial optimization problem. A closed form solution to this problem is given by the following theorem.
\begin{theorem}
Given an observation window of size $P$, let $\mathcal{D}_1$ and $\mathcal{D}_2$ be the sets defined as follows
\begin{equation*}
\mathcal{D}_1=\big\{d|P : d\leq\sqrt{P}\big\},\quad\mathcal{D}_2=\big\{d|P : d\geq\sqrt{P}\big\}.
\end{equation*}
Then, the optimum values for $N_1$ and $N_2$ are given by
\begin{align}
\begin{split}
&N_1 =\max (\mathcal{D}_1)-1,\, N_2 =\min(\mathcal{D}_2), \\
&N_1 =\min (\mathcal{D}_2)-1,\, N_2 =\max(\mathcal{D}_1).
\end{split}
\end{align}
\label{theo:mintrans}
\end{theorem}

\begin{proof}
See Appendix \ref{app:theo1}.
\end{proof}

Theorem \ref{theo:mintrans} states that in the general case there are two optimal solutions (see Fig. \ref{fig:variations}). Note, however,  that although both solutions offer the same minimal number of transmissions, they are not equivalent. A nested transmission scheme for given $N_1$ and $N_2$ creates $N_2-1$ gaps of size $N_1$ where no Doppler pulse is sent and can be used for B-mode. Therefore, the choice of $N_1$ and $N_2$ has an influence on the B-mode imaging, leading to a trade-off depending on the specific application. For example, in coherent  plane-wave compounding \cite{montaldo2009coherent}, the size of the gap determines the number of inclination angles (i.e., image quality) while the number of gaps affects the image frame rate.  

 \begin{figure}[h]
 \centering
\begin{tabular}{ |t|c|c|c|s|s|c|c|c|} 
 \hline
$N_2$   & 1     &  2   & 4   & ${\bf 8}$   & ${\bf16}$ & 32 & 64 & 128 \\ \hline
$N_1$  & 127 & 63  & 31 & ${\bf 15}$ & ${\bf 7}$  & 3   & 1    & 0 \\ \hline 
$N$       & 128 & 65  & 35 & ${\bf23}$  &${\bf 23}$ & 35 & 65 & 128 \\ \hline
\end{tabular}
\caption{{\bf Nested Array Variations.} A summary of different variations of nested arrays for an observation window of size $P=128$. The two optimal solutions are highlighted in red. } 
\label{fig:variations}
\end{figure}

In the case where $\sqrt{P}$ is an integer we get that $\max (\mathcal{D}_1)=\min (\mathcal{D}_2)=\sqrt{P}$, leading to the following corollary:

\begin{corollary}
Assuming $\sqrt{P}\in\mathbb{N}^+$,  problem (\ref{eq:mintrans}) has a unique solution. The minimal number of Doppler pulse emissions and the optimum values for $N_1$ and $N_2$ are given by
\begin{equation}
N = 2\sqrt{P}-1,\quad N_1 = \sqrt{P}-1, N_2=\sqrt{P}.
\end{equation}
\label{corollary:mintrans}
\end{corollary}

Theorem \ref{theo:mintrans} along with Corollary \ref{corollary:mintrans} imply that when an observation window with size $P$ is required, the blood power spectrum can be reconstructed from only $\Theta(\sqrt{P})$ Doppler pulse emissions. For example, given an observation window with $P=256$, perfect spectrum recovery can be achieved from 31 Doppler transmissions, which is only $12\%$ of the number of pulses sent in a standard transmission scheme. This reduction in the number of transmissions is greater than any number  previously proposed by state-of-the-art methods.

\section{Reconstruction Methods}
\label{sec:recon}
We now consider two methods to reconstruct the blood power spectrum from sub-Nyquist slow-time samples obtained using the nested transmission scheme described in (\ref{eq:khtrao}). We begin by introducing practical considerations into our framework.

First, we need to compute the autocorrelation matrix from which the subsequent signal model is derived. We estimate it by averaging samples over neighboring depths
\begin{equation}
 {\bf \hat{R}}_{{\bf y}_\text{N}}=\sum_{k=1}^{Q} {\bf y}_\text{N}[k]{\bf y}_\text{N}^H[k],
 \label{eq:estconv}
\end{equation}
where $Q$ is proportional to $f_s/f_0$.
Since the signal covariance matrix is estimated from a finite number of snapshots $Q$, the Khatri-Rao product in (\ref{eq:khtrao}) is only an approximation. Moreover, we consider additive noise to the measurements, thus,  we modify (\ref{eq:main}) to
\begin{equation}
{\bf y}_\text{N}[k]={\bf A}_\text{N}\boldsymbol{\alpha}+{\bf w}[k],
\end{equation} 
where ${\bf w}[k]\in\mathbb{C}^{N\times 1}$ is zero mean white complex Gaussian noise with unknown covariance matrix $\sigma^2\bf I$, uncorrelated with the blood scatterers amplitudes. In this case, (\ref{eq:covmtx}) and (\ref{eq:khtrao}) become
\begin{gather}
{\bf \hat{R}}_{{\bf y}_\text{N}}\approx{\bf A_\text{N}R_{\boldsymbol{\alpha}}A}_\text{N}^H +\sigma^2{\bf I}_{N\times N},\\
{\bf r}_\text{N}=\text{vec}({\bf \hat{R}}_{{\bf y}_\text{N}})\approx\tilde{\bf A}{\bf p}+\sigma^2\text{vec}({\bf I}_{N\times N}). \label{eq:main2}
\end{gather}

Next, due to the repetition of elements of $\mathcal{D}$, we have redundancy in the system of equations defined in (\ref{eq:main2}), namely, some of the rows of $\bf A$ are identical. To reduce the system of equations and the effect of noise, we define for every $d\in\mathcal{D}_u$ the set $\mathcal{M}_d$ that collects all the indices where $d$ occurs in $\mathcal{D}$
\begin{equation}
\mathcal{M}_d=\{i\,|\,\mathcal{D}(i)=d\}.
\end{equation}
Then, we define a new vector ${\bf z}\in\mathbb{C}^{(2P-1)\times 1}$ given by
\begin{equation}
{\bf z}({i_d})=\frac{1}{|\mathcal{M}_d|}\sum_{i\in\mathcal{M}_d}{\bf r}_\text{N}(i),\quad d\in\mathcal{D}_u,
\label{eq:newmes}
\end{equation}   
where $i_d$  denotes the index of $d$ in $\mathcal{D}_u$ and $|\mathcal{M}_d|$ is the cardinality of $\mathcal{M}_d$, namely, the number of times $d$ occurs in $\mathcal{D}$. Writing (\ref{eq:newmes}) in vector form, we have
\begin{equation}
{\bf z} = \bar{{\bf A}}{\bf p}+\sigma^2\bar{\bf e},
\label{eq:basic}
\end{equation} 
where $\bar{{\bf e}}\in\mathbb{R}^{(2P-1)\times 1}$ is all zeros except a 1 at the $P$th position. The matrix $ \bar{{\bf A}}\in\mathbb{C}^{(2P-1)\times M}$ has entries $\bar{{\bf A}}(d,m)=\exp(2\pi jf_mp_dT)$  where $p_d$ is the $d$th element of $\mathcal{D}_u$. To solve (\ref{eq:basic}), we present two techniques which recover the blood spectrum $\bf p$. 

\subsection{Discrete Recovery}
\label{subsec:disc}
Suppose, as in standard Doppler methods, we limit ourselves to the Nyquist grid so that $f_mT=i_m/\tilde{P}$ for every $1\leq m\leq M$, where $i_m$ is an integer in the range $0\leq i_m \leq \tilde{P}-1$ and $\tilde{P}=2P-1$. Note that our grid is twice as dense as standard Doppler techniques so that our resolution is increased by a factor of 2.  In this case, $\bar{{\bf A}}={\bf F}^H\in\mathbb{C}^{\tilde{P}\times\tilde{P}}$ where $\bf F$ is the FFT matrix and we have
\begin{equation}
{\bf z} = {\bf F}^H{\bf p}+\sigma^2\bar{\bf e}.
\label{eq:fourier}
\end{equation}
By taking the Fourier transform of (\ref{eq:fourier}) scaled by $\tilde{P}$ and using the fact that ${\bf FF}^H=\tilde{P}{\bf I}$, we obtain
\begin{equation}
\tilde{\bf z}=\frac{1}{\tilde{P}}{\bf Fz} ={\bf p}+\frac{\sigma^2}{\tilde{P}}{\bf 1},
\label{eq:fftonz}
\end{equation}
where ${\bf 1}\in\mathbb{R}^{\tilde{P}\times 1}$ is a vector of all ones. 

Finally, we adopt ideas from denoising schemes presented in \cite{pal2014soft,donoho1995noising,pal2018guarantees} and employ a soft thresholding operator $\Gamma_\lambda(x)\triangleq \max(x-\lambda,0)$ on the spectral estimates, which decreases the noise variance and the  effect of spurious frequencies resulting from the finite sample averaging. Thus, our estimate of the blood spectrum is given by
\begin{equation}
\hat{\bf p}=\Gamma_\lambda(\tilde{\bf z}),  
\label{eq:finalest}
\end{equation}
where $\lambda\geq 0$ is determined empirically and can be tuned in real-time according to the clinician's desire.
The proposed technique is outlined in Algorithm \ref{alg:nest} and is referred to as Nested Slow-Time (NEST). 

Note that NEST differs from the estimator proposed in \cite{jensen2006spectral} since NEST is based on the nested transmission scheme. Namely, the subsampling strategy is crucial for successful recovery and not only the estimate itself. Furthermore, NEST consists of additional denoising step given by soft-thresholding, which leads to a better estimate of the autocorrelation function.
\begin{algorithm}
\caption{{\bf NE}sted {\bf S}low-{\bf T}ime (NEST)}
\label{alg:nest}
\begin{algorithmic}
\REQUIRE Nested samples $\{{\bf y}[k]\}_{k=1}^Q$, threshold $\lambda\geq 0$.
\STATE {\bf 1:}  Estimate ${\bf \hat{R}}_{{\bf y}_\text{N}}$ by (\ref{eq:estconv}).
\STATE {\bf 2:}  Form ${\bf r}_\text{N}=\text{vec}({\bf \hat{R}}_{{\bf y}_\text{N}})$. 
\STATE {\bf 3:}  Compute $\bf z$ using (\ref{eq:newmes}).
\STATE {\bf 4:}  Apply a Fourier transform:  $\tilde{\bf z}=\frac{1}{\tilde{P}}\bf Fz$ with $\tilde{P}=2P-1$.
\STATE {\bf 5:}  Apply soft-thresholding: ${\bf p}=\Gamma_\lambda\big(\tilde{\bf z}\big)$.
\ENSURE $\bf p$  - Blood power spectrum. 
\end{algorithmic}
\end{algorithm}

Given $N$, the complexity of NEST is $O(N^2Q+P\log P)$. For the minimal slow-time sampling rate $N^2\propto P$ the complexity is $O(PQ+P\log P)$, making NEST suitable for real-time implementation on commercial systems.

The properties of the difference set $\mathcal{D}_u$  are emphasized in NEST. In particular, the fact that $|\mathcal{D}_u|=2P-1$ allows to achieve spectral estimates with increased resolution of $2\pi/(2P-1)T$, almost twice the resolution of standard processing. Moreover, since the elements of $\mathcal{D}_u$ are a filled ULA, the matrix $\bar{\bf A}$ reduces to a  full FFT matrix, leading to an efficient implementation.

\subsection{Continuous Recovery}
\label{subsec:cont}
In reality, grid-based methods exhibit estimation errors since the true Doppler frequencies are unlikely to lie on a predefined grid, regardless of how finely it is defined \cite{chi2011sensitivity,pal2014grid}. To address this issue, we next provide a continuous recovery method which does not assume an underlying grid. This technique is based on the work in \cite{pal2014grid,pal2014gridless} and depends on the eigenspace of the covariance matrix. Following \cite{liu2015remarks}, we construct a matrix $\widetilde{\bf R}$ given by the following theorem, which shares the same eigenspace as the covariance matrix.

\begin{theorem}
Let $\widetilde{\bf R}$ be the following Toeplitz matrix
\begin{equation}
\widetilde{\bf R}\triangleq
\begin{pmatrix}
{\bf z}(P) & {\bf z}(P-1) & \dots & {\bf z}(1) \\
{\bf z}(P+1) & {\bf z}(P) & \dots & {\bf z}(2) \\
\vdots & \vdots & \ddots & \vdots \\
{\bf z}(2P-1) & {\bf z}(2P-2) & ... & {\bf z}(P)
\end{pmatrix}. 
\label{eq:Rmtx}
\end{equation}
For an infinite number of snapshots, the matrix $\widetilde{\bf R}$ can be expressed as
\begin{equation*}
\widetilde{\bf R}= {\bf AR_{\boldsymbol{\alpha}}A}^H+\sigma^2{\bf I}_{P\times P},
\end{equation*}
where  $\bf A$ and ${\bf R_{\boldsymbol{\alpha}}}$  are defined in (\ref{eq:y=Ax}) and (\ref{eq:R=ARA}) respectively.
\label{theo:mtxstruct}
\end{theorem}

\begin{proof}
See \cite{liu2015remarks}.
\end{proof}

Note that in practice we have a finite number of snapshots, hence, the structure of $\widetilde{\bf R}$ given by Theorem \ref{theo:mtxstruct} holds only approximately. Nevertheless, from Theorem \ref{theo:mtxstruct} it follows that in the absence of noise, the range space of $\widetilde{\bf R}$ is identical to that of $\bf A$. This special structure can be exploited to recover the Doppler frequencies by using subspace methods \cite{eldar2015sampling}. We now briefly describe the ESPRIT algorithm \cite{roy1989esprit,eldar2015sampling}, provided as a representative of subspace approaches.

Assuming $M$ is known, let ${\bf E}_M$ denote the matrix of size $P\times M$ consisting of the eigenvectors corresponding to the $M$ largest eigenvalues of $\widetilde{\bf R}$. Since the matrices $\bf A$ and ${\bf E}_M$ span the same space, there exists an invertible $M\times M$ matrix  $\bf T$ such that   
\begin{equation}
{\bf A}={\bf E}_M{\bf T}.
\label{eq:AET}
\end{equation}
Let ${\bf V}_1$ be the $P-1\times M$ matrix consisting of the first $P-1$ rows of $\bf A$, and let ${\bf V}_2$ be the $P-1\times M$ matrix consisting of the last $P-1$ rows of $\bf A$. Then, we have that
\begin{equation}
{\bf V}_2={\bf V}_1{\bf \Lambda},
\label{eq:V2V1}
\end{equation}
where ${\bf \Lambda}\in\mathbb{C}^{M\times M}$ is a diagonal matrix with entries ${\bf \Lambda}(m,m)=\exp(2\pi jf_mT)$. 
In addition, let ${\bf E}_1$ and ${\bf E}_2$ be equal to the first and last $P-1$ rows of ${\bf E}_M$ respectively.  From (\ref{eq:AET}), we get
\begin{align}
\begin{split}
{\bf V}_1={\bf E}_1{\bf T}, \\
{\bf V}_2={\bf E}_2{\bf T}.
\label{eq:VET}
\end{split}   
\end{align}
Combining (\ref{eq:V2V1}) and (\ref{eq:VET}) leads to the following relation between the matrices ${\bf E}_1$ and ${\bf E}_2$:
\begin{equation}
{\bf E}_2={\bf E}_1{\bf T\Lambda T}^{-1}.
\label{eq:E2}
\end{equation}
Assuming $M\leq P-1$, the matrix ${\bf E}_1$ is full column rank, therefore, ${\bf E}_1^\dagger{\bf E}_1={\bf I}$ where ${\bf E}_1^\dagger$ is the pseudo-inverse of ${\bf E}_1$. Multiplying (\ref{eq:E2}) on the left by ${\bf E}_1^\dagger$ leads to 
\begin{equation}
{\bf E}_1^\dagger{\bf E}_2={\bf T\Lambda T}^{-1}.
\label{eq:E1E2}
\end{equation}  
Following  (\ref{eq:E1E2}), we can recover the Doppler frequencies from the eigenvalues of ${\bf E}_1^\dagger{\bf E}_2$. 

ESPRIT requires knowledge of the number of Doppler frequencies $M$, which is typically unavailable to us. In practice, one can estimate $M$ using, for example, the minimum description length (MDL) algorithm \cite{roy1989esprit}. Here, we propose an alternative based on low rank approximation \cite{pal2014gridless}.

Let the eigen-decomposition of $\widetilde{\bf R}$ be given by
\begin{equation}
[{\bf E},{\bf d}]=\text{eig}(\widetilde{\bf R}),
\end{equation}
where $\bf E$ consists of the eigenvectors in its columns and $\bf d$ is a vector consisting of the eigenvalues in a non-increasing order. To promote low rank of the matrix $\widetilde{\bf R}$, we perform soft-thresholding on $\bf d$ and estimate $M$ as 
\begin{equation}
M=||\Gamma_\lambda({\bf d})||_0,
\end{equation}
where $\lambda\geq0$ is chosen empirically and $||\cdot||_0$ is the $l_0$ semi-norm which counts the number of nonzero elements of the vector.
This operation acts as a denoising scheme and accounts for the finite snapshot effect on the estimates. Given the estimate of $M$, we define ${\bf E}_M$ as the first $M$ columns of $\bf E$ and perform ESPRIT as described. 

Once the Doppler frequencies are recovered, the Vandermonde matrix $\bf \bar{A}$, defined in (\ref{eq:basic}), is constructed. Assuming $2P>M$ the matrix $\bar{\bf A}$ has full column rank and the blood spectrum vector $\bf p$ is then obtained by left inverting $\bar{\bf A}$,
\begin{equation}
\hat{\bf p}=\bar{\bf A}^\dagger{\bf z}.
\end{equation}
 The proposed recovery method is summarized in Algorithm\,\ref{alg:nesprit} and is referred to as NESPRIT.
 \begin{algorithm}
\caption{{\bf NE}sted {\bf S}low-Time ES{\bf PRIT} (NESPRIT)}
\label{alg:nesprit}
\begin{algorithmic}
\REQUIRE Nested samples $\{{\bf y}[k]\}_{k=1}^Q$, threshold $\lambda\geq 0$.
\STATE {\bf 1 :}  Estimate ${\bf \hat{R}}_{{\bf y}_\text{N}}$ by (\ref{eq:estconv}).
\STATE {\bf 2 :}  Form ${\bf r}_\text{N}=\text{vec}({\bf \hat{R}}_{{\bf y}_\text{N}})$. 
\STATE {\bf 3 :}  Compute $\bf z$ using (\ref{eq:newmes}).
\STATE {\bf 4 :}  Construct $\widetilde{\bf R}$ according to (\ref{eq:Rmtx}).
\STATE {\bf 5 :}  Decompose $\widetilde{\bf R}$ :  $[{\bf E},{\bf d}]=\text{eig}(\widetilde{\bf R})$.
\STATE {\bf 6 :} Estimate  $M=||\Gamma_\lambda(d)||_0$.
\STATE {\bf 7 :} Extract ${\bf E}_M=[{\bf e}_1,\dots,{\bf e}_M]$.
\STATE {\bf 8 :} Define ${\bf E}_1 \text{  and } {\bf E}_2$ as in (\ref{eq:VET}).
\STATE {\bf 9 :} Compute the eigenvalues of ${\bf E}_1^\dagger{\bf E}_2$: $\boldsymbol{\beta}=\text{eig}({\bf E}_1^\dagger{\bf E}_2)$.
\STATE {\bf 10:} Estimate the Doppler frequencies $\boldsymbol{f} = \frac{\angle\boldsymbol{\beta}}{2\pi T}$
\STATE {\bf 11:} Construct $\bar{\bf A}$ defined in (\ref{eq:basic}) using $\boldsymbol{f}$.
\STATE {\bf 12:} Spectrum recovery: ${\bf p}=\bar{\bf A}^\dagger{\bf z}$
\ENSURE $(\boldsymbol{f},{\bf p})$ - Blood power spectrum.
\end{algorithmic}
\end{algorithm}

The NESPRIT algorithm can theoretically exhibit infinite frequency-precision in identifying the Doppler frequencies when there is no noise. However, it has a large computational load. The complexity of NESPRIT is dominated by the eigen-decomposition of a $P\times P$ Hermitian matrix, which requires  $O(P^3)$ operations \cite{pan1999complexity}. Note, however, that more computationally efficient methods, presented in  \cite{xu1994fast}, may be used to reduce the complexity of traditional ESPRIT.

\subsection{Clutter Filtering and Apodization}
One major challenge in spectral Doppler is clutter filtering. Clutter signals stem from backscattered echoes from vessel’s walls and surrounding tissues, stationary and non-stationary, and are typically 40 to 60 dB stronger than the flow signal \cite{lediju2009sources,jensen1996estimation,hedrick1995ultrasound}. Thus, clutter may obscure blood velocities and must be removed for accurate velocity estimation. 

Conventionally, clutter removal is applied using high-pass finite impulse response (FIR) filters or infinite impulse response (IIR) filters. However, such filters assume uniformly sampled data, which is not the case when using sparse Doppler sequences. 
To overcome this, in \cite{jensen2006spectral,gudmundson2011blood}, polynomial regression filters were used for clutter rejection since they are not restricted to uniform sampling. The downside of regression filters is that they may lead to spurious frequencies in the output spectrum \cite{avdal2015effects,torp1997clutter}, compromising their reliability for clinical use.    

A crucial disadvantage of many sparse Doppler methods is their inability to use FIR and IIR filters for clutter removal. Fortunately, NEST and NESPRIT do not share this limitation, since they recover  the full uniform autocorrelation function, allowing to perform filtering in the correlation domain as we show next.  

Consider a linear time invariant (LTI) stable system with impulse response $h[n]$, driven by a WSS discrete process $x[n]$.  Denoting by $y[n]$ the output of the system and by $R_y[n]$   the autocorrelation function of $y[n]$, we have
\begin{align}
\begin{split}
y[n] &= x[n]\ast h[n] \\
R_y[n] &= R_x[n]\ast h[n]\ast h[-n],
\label{eq:filter}
\end{split}
\end{align}
where $R_x[n]$ is the autocorrelation function of the input and $\ast$ denotes convolution. Following (\ref{eq:filter}), any FIR or IIR filter $h[n]$ can be applied in the correlation domain by computing
\begin{equation}
\tilde{\bf z}[n]={\bf z}[n]\ast h[n] \ast h[-n],
\end{equation}  
where ${\bf z}[n]$ is given by (\ref{eq:newmes}). Thus, the fact that we recover the full uniform autocorrelation function allows us to perform clutter removal using any desired filter. In addition, specifically for NESPRIT,  which  involves an eigenvalue decomposition, eigen-based clutter filters \cite{alfred2010eigen} are directly applicable.

Similarly, any apodization  function $a[n]$, used for reducing sidelobes,  may be applied directly in the correlation domain by computing
\begin{equation}
\hat{\bf z}[n]= {\bf z}[n] \cdot R_a[n],{}
\end{equation} 
where ${\bf z}[n]$ is given by (\ref{eq:newmes}) and $R_a[n]$ is the autocorrelation function of $a[n]$.

\section{Alternative Sparse Arrays}
\label{sec:alternatives}
So far, we considered only nested arrays as an approach for reducing the number of Doppler transmissions. However, in the literature of array processing there are alternative sparse array configurations which can match the performance of their fully populated counterparts. In this section, we briefly review several alternatives and discuss their properties in comparison with nested arrays. 

\subsection{Super Nested}
A modified version of nested arrays are the super nested arrays \cite{liu2016superconf,liu2016super,liu2016high}. Assuming $N_1\geq 4$ and $N_2\geq 3$, super nested arrays are specified by the integer set $\mathcal{S}_\text{SN}$ created by concatenating six ULAs (see Fig. \ref{fig:patterns2}), defined by  
\begin{align}
\begin{split}
\mathcal{S}_\text{SN}&=\mathcal{X}_1\cup\mathcal{Y}_1\cup\mathcal{X}_2\cup\mathcal{Y}_2\cup\mathcal{Z}_1\cup\mathcal{Z}_2, \\
\mathcal{X}_1&= \{1+2l\,|\,0\leq l\leq A_1\},  \\
\mathcal{Y}_1&= \{(N_1+1)-(1+2l)\,|\,0\leq l\leq B_1\}, \\ 
\mathcal{X}_2&= \{(N_1+1)+(2+2l)\,|\,0\leq l\leq A_2\}, \\
\mathcal{Y}_2&= \{2(N_1+1)-(2+2l)\,|\,0\leq l\leq B_2\}, \\
\mathcal{Z}_1&= \{l(N_1+1)\,|\,2\leq l\leq N_2\}, \\
\mathcal{Z}_2&= \{N_2(N_1+1)-1\} \\
\end{split}
\end{align} 
with
\small
\begin{align*}
(A_1,B_1,A_2,B_2)=
\begin{cases}
(r,r-1,r-1,r-2),\quad &N_1=4r,\\
(r,r-1,r-1,r-1),\quad &N_1=4r+1,\\
(r+1,r-1,r-1,r-2),\quad &N_1=4r+2,\\
(r,r,r,r-1),\quad &N_1=4r+3,
\end{cases}
\end{align*}
\normalsize
where $r$ is an integer.

\begin{figure}[h]
\centering
\includegraphics[width=1\linewidth]{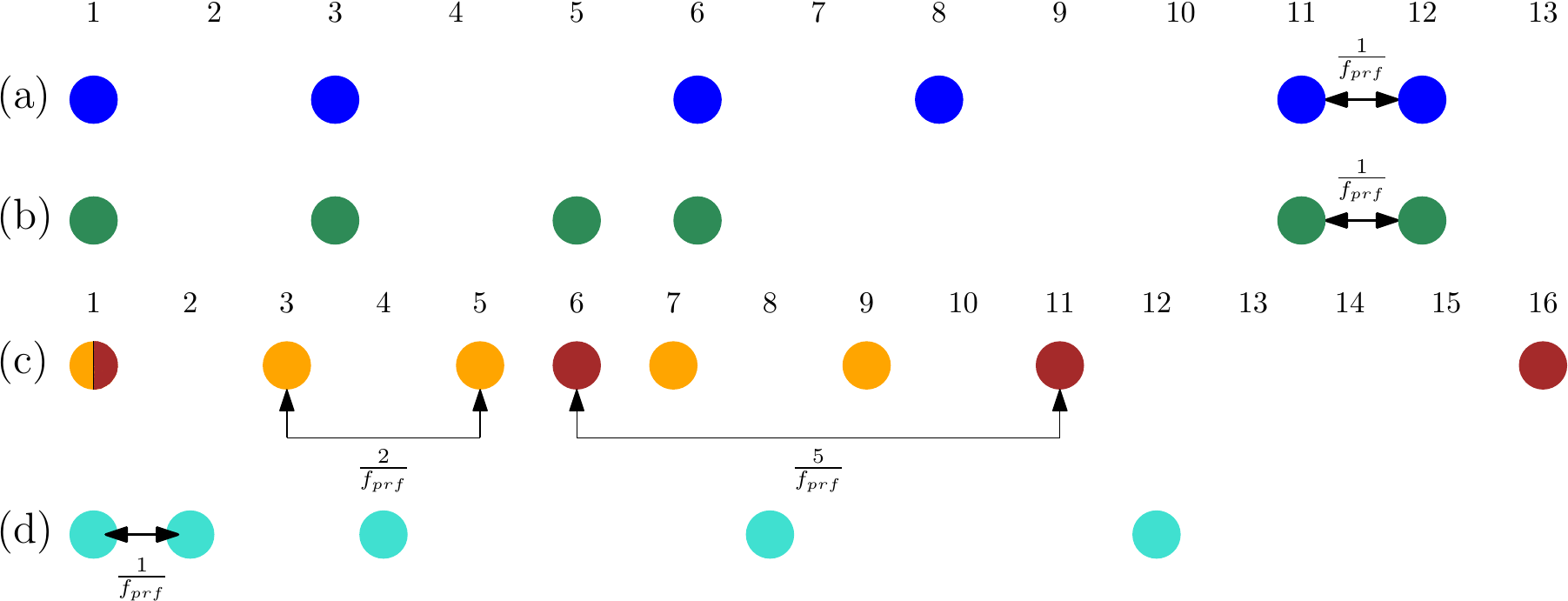}
\caption{{\bf Alternative Transmission Patterns.} Different transmission patterns for various observation windows. (a) Super nested pattern for $N_1=N_2=3,\,P=12$. (b) 3rd order super nested pattern for $N_1=N_2=3,\,P=12$. (c) Co-prime scheme for $N_1=2,N_2=5,\,P=11$, where a two color circle represents a single Doppler transmission that is mutual for both sub-arrays. (d) 3-level nested array for $N_1=N_2=1,\,N_3=3,\,P=12$. Matlab code for generating super nested arrays can be found in \cite{superMatlab}.} 
\label{fig:patterns2}
\end{figure}


These variants share the same properties as nested arrays in terms of the number of Doppler transmissions and their difference sets. In addition, they offer an advantage over nested arrays of reduced mutual coupling \cite{merrill2001introduction}, which in our case translates to the effect of previous transmissions on the received signal corresponding to the current emission. This property may allow increasing the maximal depth examined. However, super nested arrays exhibit complex geometry compared to nested arrays. In particular, the Doppler gaps created are not of the same size and thus using them for B-mode imaging may be difficult in certain applications.

\subsection{Co-Prime Array}
This type of sparse array has been studied extensively in the literature \cite{vaidyanathan2013coprime,vaidyanathan2011sparse,pal2011coprime,vaidyanathan2012direct,pal2012correlation,tan2014direction,vaidyanathan2011theory}. Let $N_1<N_2$ be co-prime integers, i.e., their greatest common divisor (gcd) is 1. A co-prime array is composed of two ULAs with inter-element spacing $N_1$ and $N_2$:
\begin{align}
\begin{split}
&\mathcal{S}_{\text{N}_1} = \{n_1N_2,\quad n_1=0,1,...,2N_1-1\}, \\
&\mathcal{S}_{\text{N}_2} = \{n_2N_1,\quad n_2=0,1,...,N_2-1\}, \\
&\mathcal{S}_\text{CP} = \{\mathcal{S}_{\text{N}_1}\cup\mathcal{S}_{\text{N}_2}\}.
\end{split}
\end{align}

By Lemma 1 in \cite{pal2011coprime}, the difference set of $\mathcal{S}_\text{CP}$ contains all $2N_1N_2+1$ contiguous integers from $-N_1N_2$ to $N_1N_2$. This means that for the choice of $N_1$ and $N_2$ such that $N_1N_2=P-1$, we can recover all time lags of the autocorrelation continuously from $-(P-1)$ to $P-1$ as in nested arrays.
In addition, a co-prime array has the property of reduced mutual coupling compared to a nested array, while having a simpler geometry compared to super nested arrays.



The main drawback of co-prime arrays is that they require sending Doppler pulses in times beyond the observation window, as can be seen in Fig. \ref{fig:patterns2}. Therefore, the reflected slow-time signal may not preserve its stationarity property, which is a key assumption in Doppler processing. To overcome this, we can limit ourselves to Doppler transmissions sent within the observation window. However, in this case, the difference set is not a filled ULA, i.e., not all time lags are recovered. As a result, this will reduce the number of DOF, namely, the number of velocities that are recoverable. 

\subsection{K-Level Nested Array}
The nested array concept is based on concatenating two ULAs. A \textit{K-Level nested array } is an extension to \textit{K} ULAs. This array is parameterized by $K,N_1,N_2,...,N_K\in\mathbb{N}^+$ and defined as follows:
\begin{align}
\begin{split}
&\mathcal{S}_1={1,2,...,N_1}, \\
&\mathcal{S}_i=\Big\{n\prod_{j=1}^{i-1}(N_j+1),\quad n=1,2,...,N_i\Big\}\quad i=2,3,...,K,\\
&\mathcal{S}_\text{KL}=\bigcup_{i=1}^K \mathcal{S}_i.
\end{split}
\end{align} 
The inter-element spacing in the $i$th level is equal to $N_{i-1}+1$ times the spacing in the
$(i-1)$th level, as illustrated in Fig.\,\ref{fig:patterns2}. 

To determine the minimal number of transmissions using this approach, we define a generalized version of problem (\ref{eq:mintrans}):
\begin{align}
\begin{split}
\underset{K\in\mathbb{N}^{+}}{\min}\quad\underset{N_1,...,N_K\in\mathbb{N}^{+}}{\min}\qquad &\sum_{i=1}^{K} N_i \\ 
\text{subject to}\qquad\qquad N_k&\prod_{i=1}^{K-1}(N_i+1)=P.
\end{split}
\label{eq:kmintrans}
\end{align}
The solution to (\ref{eq:kmintrans}) is given by the following theorem.
\begin{theorem}
Let $P$ be the size of a given observation window, represented by its prime factorization 
\begin{equation*}
P=\prod_{i=1}^{\omega}p_i^{q_i},
\end{equation*}
where $\omega$ is the number of distinct prime factors of $P$. Define $\Omega=\sum_{i=1}^{\omega}q_i$. The optimal
number of nesting levels $K$ and the minimal number of transmissions are given by
\begin{align*}
&K = \Omega, \\
&N =  1+\sum_{i=1}^{\omega}(p_i-1)q_i, \\
&\Big\{N_{i}\Big\}_{i=1}^{\Omega} =\Big\{ \underset{q_1\text{times}}{\underbrace{p_1-1,...,p_1-1}},...,\underset{q_{\omega}-1\text{ times}}{\underbrace{p_{\omega}-1,...,p_{\omega}-1}},p_{\omega}\Big\}.
\end{align*}
\label{theo:kmintrans}
\end{theorem}

\begin{proof}
See Appendix \ref{app:theo3}.
\end{proof}

A common choice for $P$ is a power of two. The optimal K-level nested array in this case is given by the following corollary:

\begin{corollary}
Consider an observation window of size $P=2^n$ for some $n\in\mathbb{N}^+$. The optimal transmission pattern consists of $n+1$ emissions with exponential spacing, given by the set
\begin{equation*}
\mathcal{S}_\text{opt}=\{1,2,4,....2^n\}.
\end{equation*} 
\end{corollary}

Nested arrays are associated with second-order statistics while K-level nested arrays extend this notion to higher-order statistics. For example, 4-level nested arrays are related to differences of the difference set, i.e. 4th-order moments. Thus, if we consider higher-order statistics, then Theorem \ref{theo:kmintrans} implies that K-level nested arrays offer a significant reduction in the number of Doppler transmissions over nested arrays. However, using higher-order statistics requires a large number  of snapshots, which may not be available. 

\section{Simulations and In Vivo Results}
\label{sec:sim}

  We now demonstrate blood spectrum reconstruction from sparse slow-time samples. The NEST and NESPRIT algorithms are evaluated using Field II \cite{jensen1996field,jensen1992calculation} simulations with  the Womersley model \cite{womersley1955xxiv} for pulsating flow from the femoral artery. The specific parameters for the Field II simulation of the flow are summarized in Table \ref{table:simparams}. The estimation of the autocorrelation matrix was performed using $Q = 33$ regularly spaced samples along depth and involved subtraction of the mean of the signal, thus removing the signal's stationary part.

\begin{table}
\centering
\begin{tabular}{l c  c  } 
 \hline
 Transducer center frequency  & $f_0$ & 3.5  [MHz]  \\ 
 Pulse repetition frequency & $f_{prf}$ & 5 [kHz]  \\
 Sampling frequency & $f_s$ & 20 [MHz]  \\
 Speed of sound & $c$ & 1540 [m/s]  \\
 Mean velocity &  & 0.1 [m/s] \\
Beam/flow angle &  & $60^o$ \\
Observation window size & $P$ & 256 \\
 \hline
\end{tabular}
\caption{Parameters for femoral flow simulation}
\label{table:simparams}
\end{table}

\subsection{MSE versus SNR}
\label{subsec:mse}
First, we evaluate the performance of the proposed algorithms by using a simplified signal simulated according to (\ref{eq:y_kp2}), comprising a single Doppler frequency which does not lie on the grid of standard processing. We consider an observation window of size $P=8$ and a nested transmission scheme where $N_1=3,N_2=2$ and $T=1$. Assuming a Doppler frequency $f=3/15=0.2$, we compare NEST, NESPRIT and Welch's method by studying the mean squared error (MSE) of their frequency estimates as a function of signal-to-noise ratio (SNR). We define the MSE of an estimate  $\hat{f}$
as 
\begin{equation}
MSE(\hat{f})=\mathbb{E}[(f-\hat{f})^2],
\end{equation}
where $\mathbb{E}[\cdot]$ is the expectation operator evaluated empirically using 1000 Monte Carlo simulations.

Figure \ref{fig:mse} shows the MSE of the three methods as a function of SNR for $Q=200$ snapshots. Notice how the performance of the three methods improves considerably with increasing SNR. In low SNR regimes NEST performs the worst while the performance of NESPRIT and Welch's method are comparable. However, while from a certain point both NEST and NESPRIT recover the Doppler frequency perfectly, Welch's method still produces an error even in the high SNR regime. This is expected due to the limited Doppler resolution of Welch's method compared to NEST and NESPRIT.

\begin{figure}[h]
\centering
\includegraphics[width=1\linewidth]{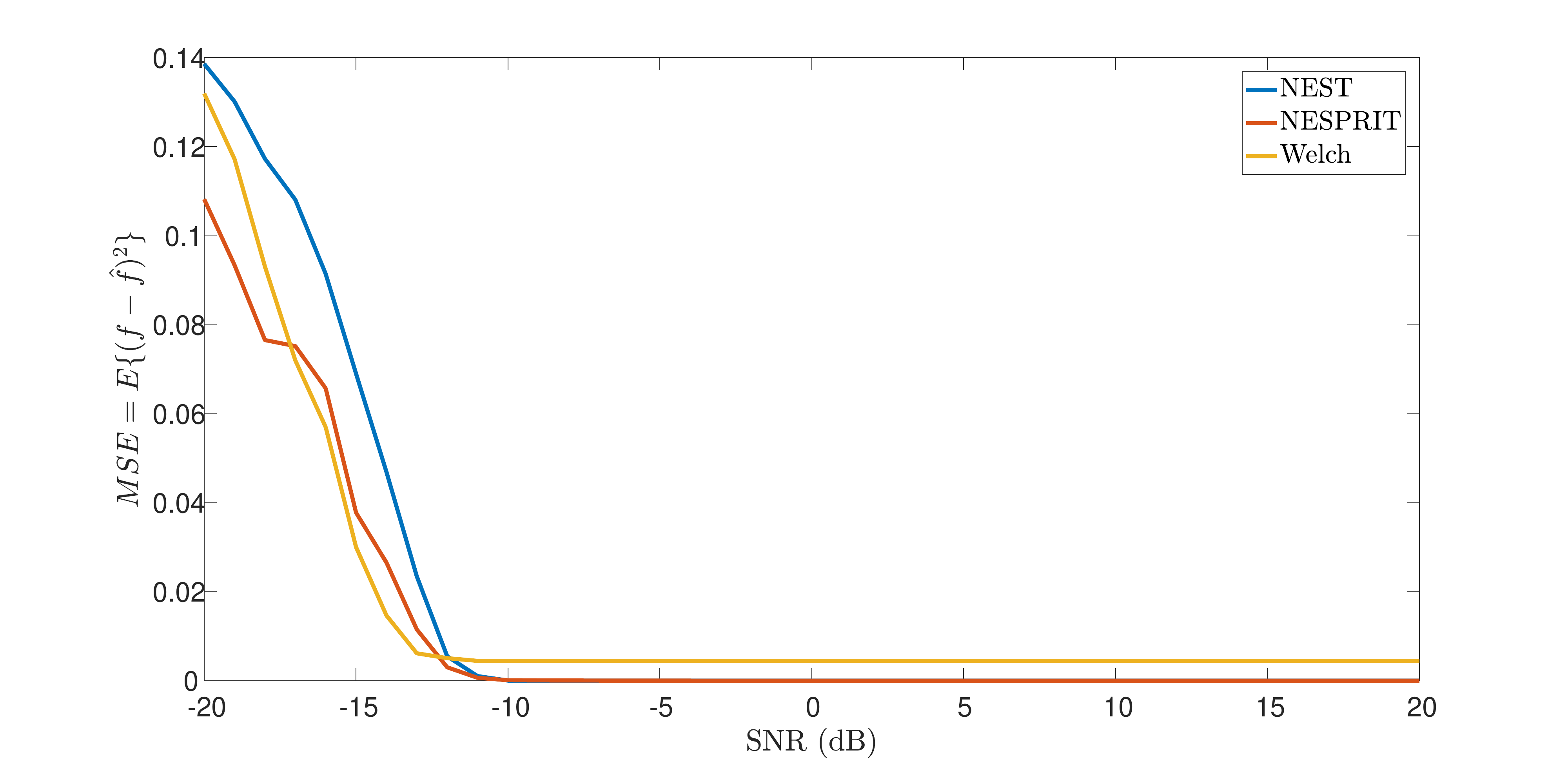}
\caption{{\bf MSE versus SNR.} MSE as a function of SNR (for a single Doppler frequency) of NEST and NESPRIT methods applied for  a nested transmission scheme with $N_1=3,\,N_2=2$ and $Q=33$.} 
\label{fig:mse}
\end{figure}

\subsection{Different Slow-Time Subsampling Levels}
\label{subsec:samplinglevel}
We now investigate the spectrum recovery  of NEST and NESPRIT using the proposed sparse transmission scheme with different levels of slow-time subsampling, i.e., different number of Doppler emissions:
\begin{enumerate}
\item $N=129\,(\approx 50.3\%): N_1=127,N_2=2$
\item $N=67\,(\approx 26.1\%)\,\,\,: N_1=63,N_2=4$
\item $N=39 \,(\approx 15.2\%)\,\,\,: N_1=31,N_2=8.$
\end{enumerate}

Figure\,\ref{fig:subsampling} shows the spectrogram of traditional Welch’s method and the ones obtained with NEST (top) and NESPRIT (bottom) using 50.3\%, 26.1\% and 15.2\% of possible Doppler transmissions. As can be seen, for all levels of subsampling both proposed algorithms produce a clear and accurate spectrogram. This allows the user the freedom to vary $N_1$ and $N_2$ and thus determine the level of subsampling dynamically.

\begin{figure*}
\centering
\includegraphics[trim={3cm 4cm 3cm 4.5cm},clip,width=1\linewidth]{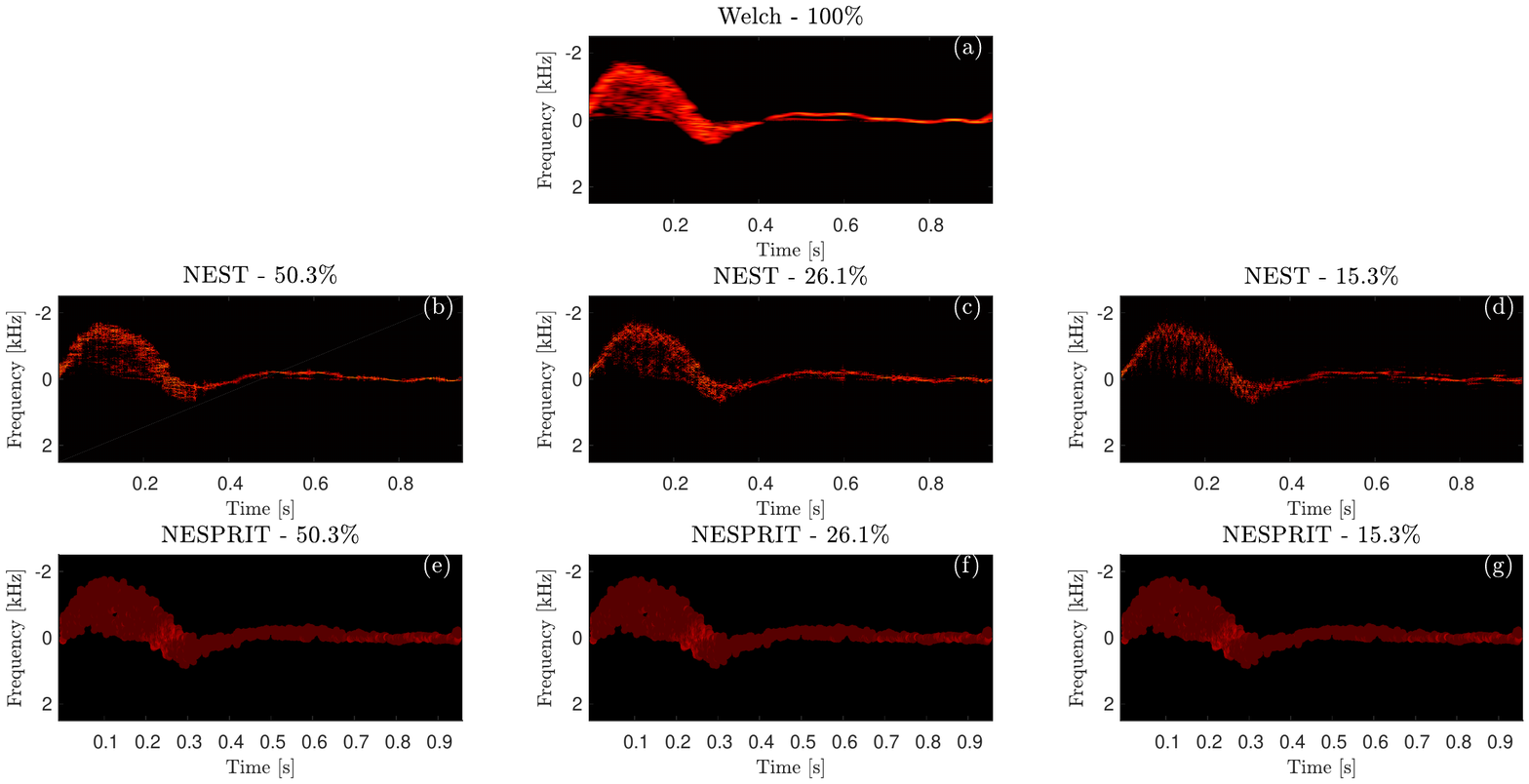}
\caption{{\bf Different Subsampling Levels.} Spectrograms of the simulated femoral artery using different slow-time subsampling from 256 pulses (100\%) down to 39 pulses ($\approx$15\%). (a) Welch's method - 100\% (b) NEST - $50.3\%$ (c) NEST - $26.1\%$ (d) NEST - $15.2\%$ (e) NESPRIT - $50.3\%$ (f) NESPRIT - $26.1\%$ (g) NESPRIT - $15.2\%$. All spectrograms are displayed with a dynamic range of 60 dB.}
\label{fig:subsampling}
\end{figure*}

\subsection{Clutter Filter and Apodization}
Next we demonstrate the application of clutter filtering and apodization using NEST and NESPRIT techniques. To that end, a clutter signal was superimposed on the flow model being 40 dB stronger than the blood signal. We use a Butterworth high pass filter with normalized cutoff frequency 0.03 and apodization with a Hamming window of length 256. Recall that these actions are performed on the autocorrelation signal given by (\ref{eq:newmes}). 

Figure \ref{fig:clutter} presents spectrograms of NEST (top) and NESPRIT (bottom) reconstructed from approximately 25\% of the Doppler emissions. On the left side the resultant unfiltered spectrograms are given. As seen, only the frequency related to the clutter signal  is visible, since the clutter obscures the blood velocities entirely. Applying a high pass filter on the autocorrelation signal produces adequate spectrograms (middle) where clearly the low frequencies are filtered out. As expected, there are artifacts due to the fact that the filtering is not ideal and part of the blood signal is also filtered out along with the clutter. Using Hamming apodization helps in reducing these artifacts, yielding cleaner spectrograms (right). 

These last results emphasize the importance of recovering the slow-time autocorrelation which allows to incorporate any conventional clutter filter and apodization in NEST and NESPRIT.  
\label{subsec:clutter}

\begin{figure*}
\centering
\includegraphics[trim={3cm 1cm 3cm 0.5cm},clip,width=1\linewidth,height=8cm]{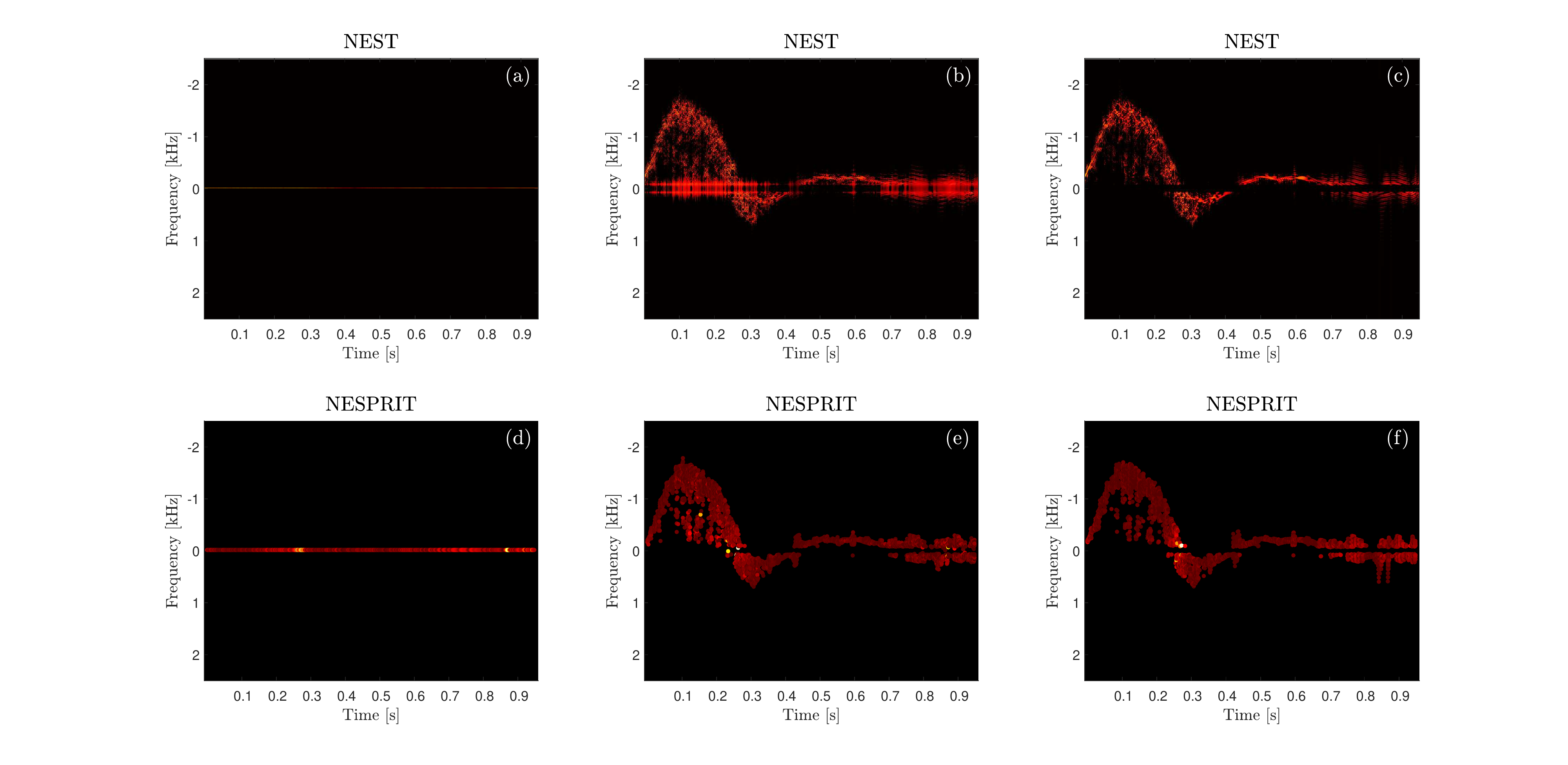}
\caption{{\bf Clutter Filtering and Apodization.} Spectrograms of the simulated femoral artery with superimposed clutter signal. (a) NEST with no filter (b) NEST with high pass filter (c) NEST with high pass filter and Hamming apodization (d) NESPRIT with no filter (e) NESPRIT  with high pass filter (f) NESPRIT with high pass filter and Hamming apodization. All spectrograms reconstructed using only 67 transmissions (25\%) and displayed with a dynamic range of 60 dB. }
\label{fig:clutter}
\end{figure*}

\subsection{Alternative Sampling Patterns}
\label{subsec:alternatives}
Here we examine other transmission schemes reviewed in Section \ref{sec:alternatives}. In Fig. \ref{fig:alternatives} the spectrograms recovered by NEST (top) and NESPRIT (bottom) are presented, where the input vector was acquired in each setting according to a different transmit pattern - super-nested (left), co-prime (middle), 4-level nested (right).  The parameters of each emission scheme are presented in Table \ref{table:alterparams}.  
As can be seen in Fig. \ref{fig:alternatives}, for co-prime and 4-level nested patterns, NEST and NESPRIT failed to produce clear spectrograms and exhibit severe artifacts. This is expected since when using these transmit schemes the resulting autocorrelations have holes, leading to aliasing which is dramatic especially when the spectrum consists of  a wide range of frequencies. Note that for the 4-level nested scheme, NESPRIT failed to produce a visible spectrogram and hence is not shown. Moreover, the spectrograms resulting from the super-nested pattern, although clear, exhibit aliasing which is surprising because the super-nested approach shares the nested pattern property of having a full autocorrelation. This aliasing is probably due to the fact that in super-nested transmission there is only one pair of transmissions separated in time by $T$, which may lead to inaccurate estimation of lag one of the autocorrelation, effectively reducing the PRF by a factor of 2.

\begin{table}
\centering
\begin{tabular}{l c } 
 \hline
 Super-nested   &  $N_1=15$ $N_2=16$  \\ 
 Co-prime           &  $N_1=14$ $N_2=9$  \\
 4-Level nested &  $N_1=N_2=N_3=3$ $N_4=4$  \\
 \hline
\end{tabular}
\caption{Parameters for different transmit patterns}
\label{table:alterparams}
\end{table}

\begin{figure*}
\centering
\includegraphics[trim={3cm 1cm 3cm 0cm},clip,width=1\linewidth,height=8cm]{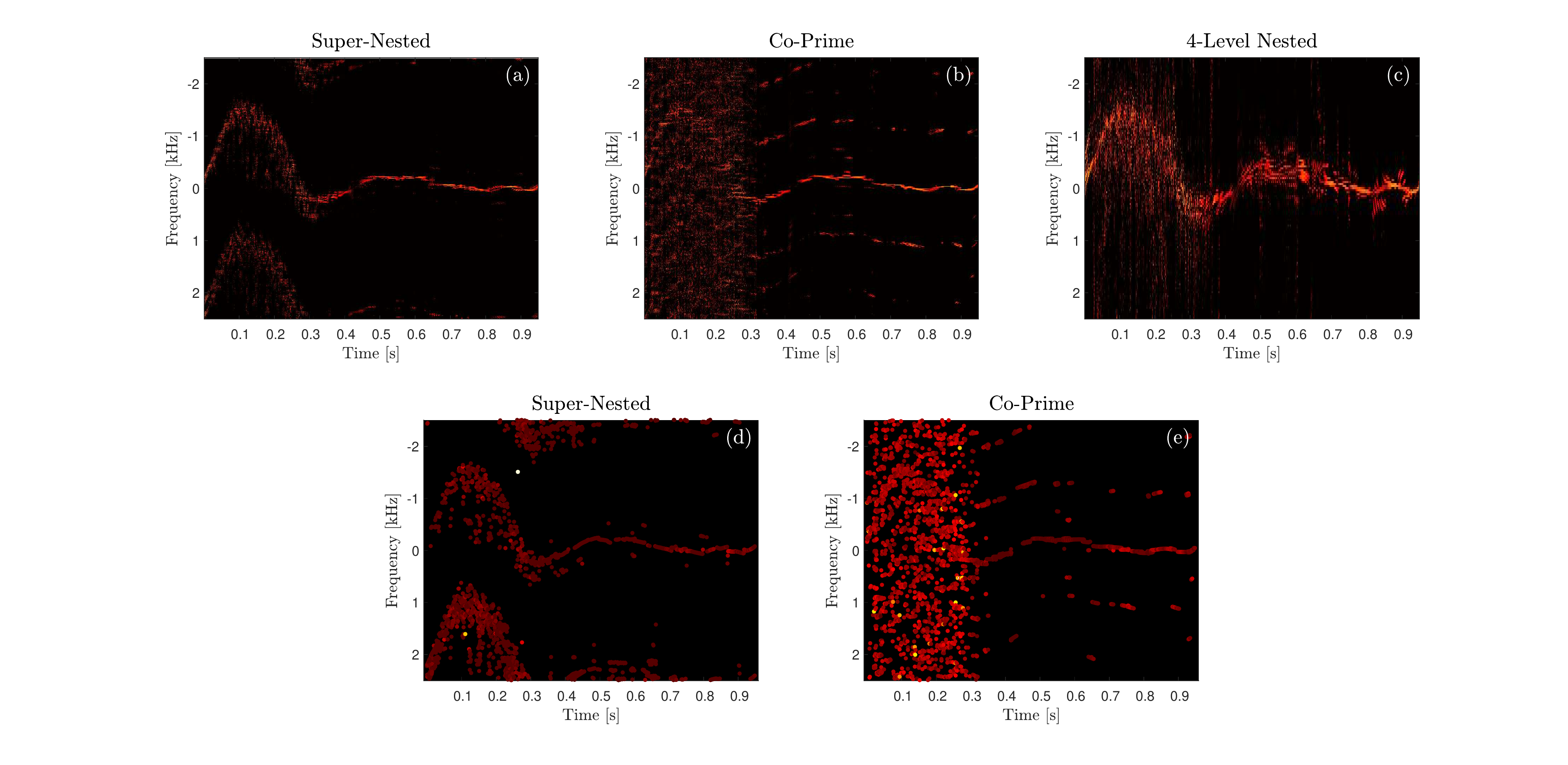}
\caption{{\bf Alternative Transmission Schemes.} Spectrograms of the simulated femoral artery using alternative emission patterns. (a) Super-Nested scheme with NEST recovery (b) Co-Prime scheme with NEST recovery (c) 4-Level Nested with NEST recovery (d) Super-Nested scheme with NESPRIT (e) Co-Prime scheme with NESPRIT recovery. All spectrograms are displayed with a dynamic range of 60 dB.}
\label{fig:alternatives}
\end{figure*}

\subsection{Minimal Rate Performance}
\label{subsec:minrate}
As a final simulation, we test the performance of both NEST and NESPRIT for the minimal slow-time sampling rate. According to the nested approach, for an observation window of size $P=256$ the minimal number of Doppler transmissions is $2\sqrt{P}-1=31$, which is 12\%  of 256. Based on this subsampling scheme, the proposed techniques are compared with the conventional Welch's method and with two recent developed techniques BSLIM and BIAA which can handle arbitrary sampling schemes of the slow-time data. The resulting spectrograms are shown in Fig. \ref{fig:minrate}. As can be seen, the blood spectrograms formed by NEST and NESPRIT are sharp and clear, whereas, Welch's method, BSLIM and BIAA produce spectrograms with significant artifacts, especially in regions of high velocities due to aliasing. These last results prove that NEST and NESPRIT, based on the proposed transmission scheme, are able to fully recover the blood spectrum only from 12\%. This along with the fact that NEST and NESPRIT present a closed form solution, in contrast to other competitive methods, indicate that NEST and NESPRIT outperform current state-of-the-art techniques. 

\begin{figure*}
\centering
\includegraphics[trim={3cm 1cm 3cm 0cm},clip,width=1\linewidth,height=8cm]{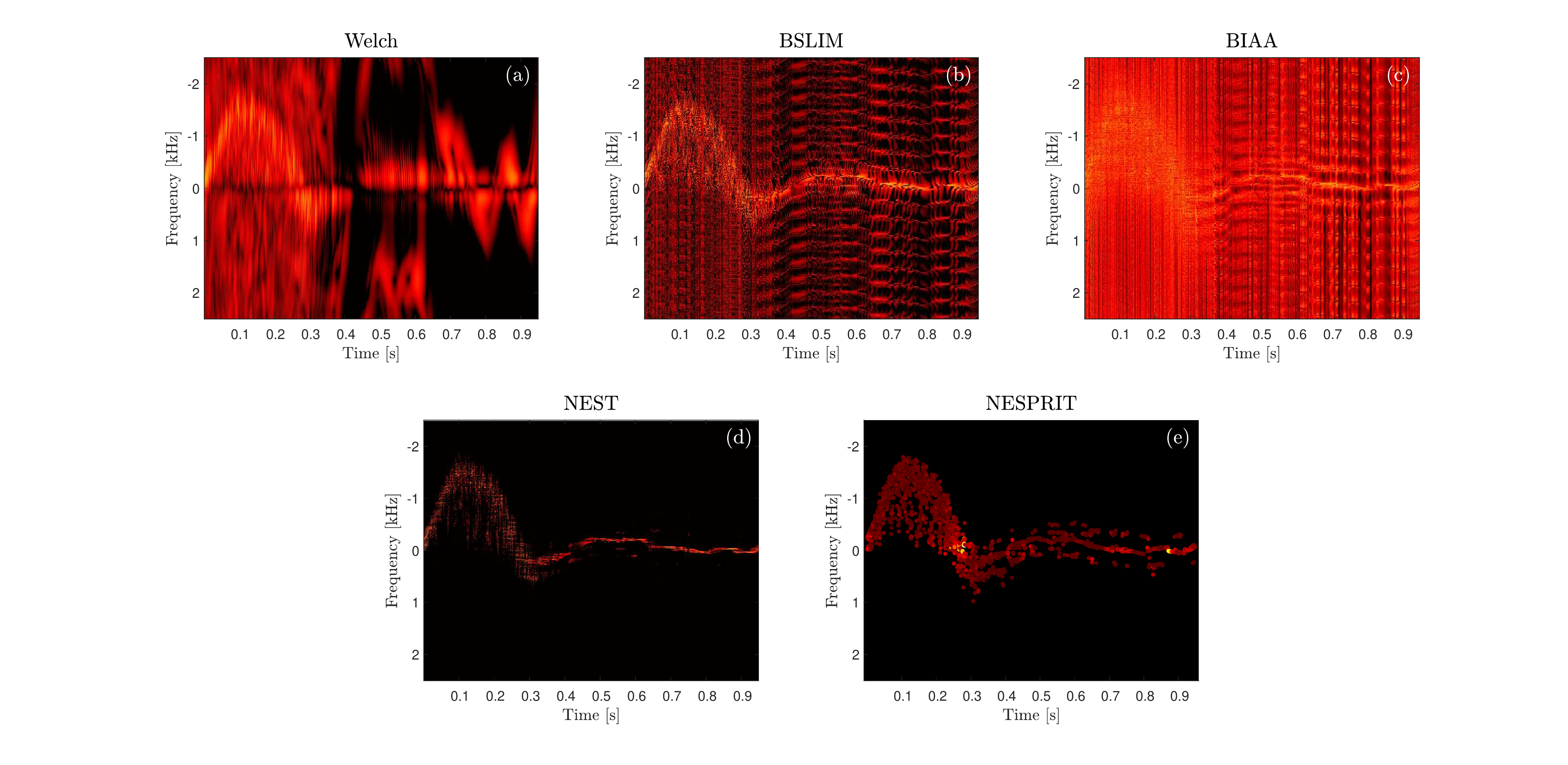}
\caption{{\bf Performance Comparison for Minimal Rate.} Spectrograms of the simulated femoral artery using only 31 pulses out of 256 (12\%) according to the sparse emission scheme. (a) Welch's method (b) BSLIM (c) BIAA (d) NEST (e) NESPRIT. All spectrograms are displayed with a dynamic range of 60 dB.}
\label{fig:minrate}
\end{figure*}

\subsection{In vivo}
\label{subsec:invivo}

We end by evaluating the performance of the proposed methods on in vivo data obtained online\footnote[1]{The data was downloaded from http://bme.elektro.dtu.dk/31545/.}. The data consists of a Carotid artery of a healthy volunteer examined using B-K 8556 ultrasound scanner with a 3.2 MHz linear array probe transducer in duplex mode. The sampling frequency was 8 kHz and the pulse  repetition frequency was 3.5 MHz. An observation window of $P=128$ samples was chosen and a nested transmission scheme with $N_1=31$ and $N_2=4$ for both NEST and NESPRIT, leading to a total number of 35 emissions ($\sim 27\%$). The obtained spectrograms are shown in Fig. \ref{fig:invivo}. As can be seen from the figure, NEST and NESPRIT successfully recover the Doppler frequencies from a small number of transmissions, producing similar spectrograms to that obtained by Welch's method using the fully sampled data. These results validate the effectiveness of the proposed methods and their potential for clinical use.    

\begin{figure*}
\centering
\includegraphics[trim={0cm 4cm 0cm 4cm},clip,width=1\linewidth,height=8cm]{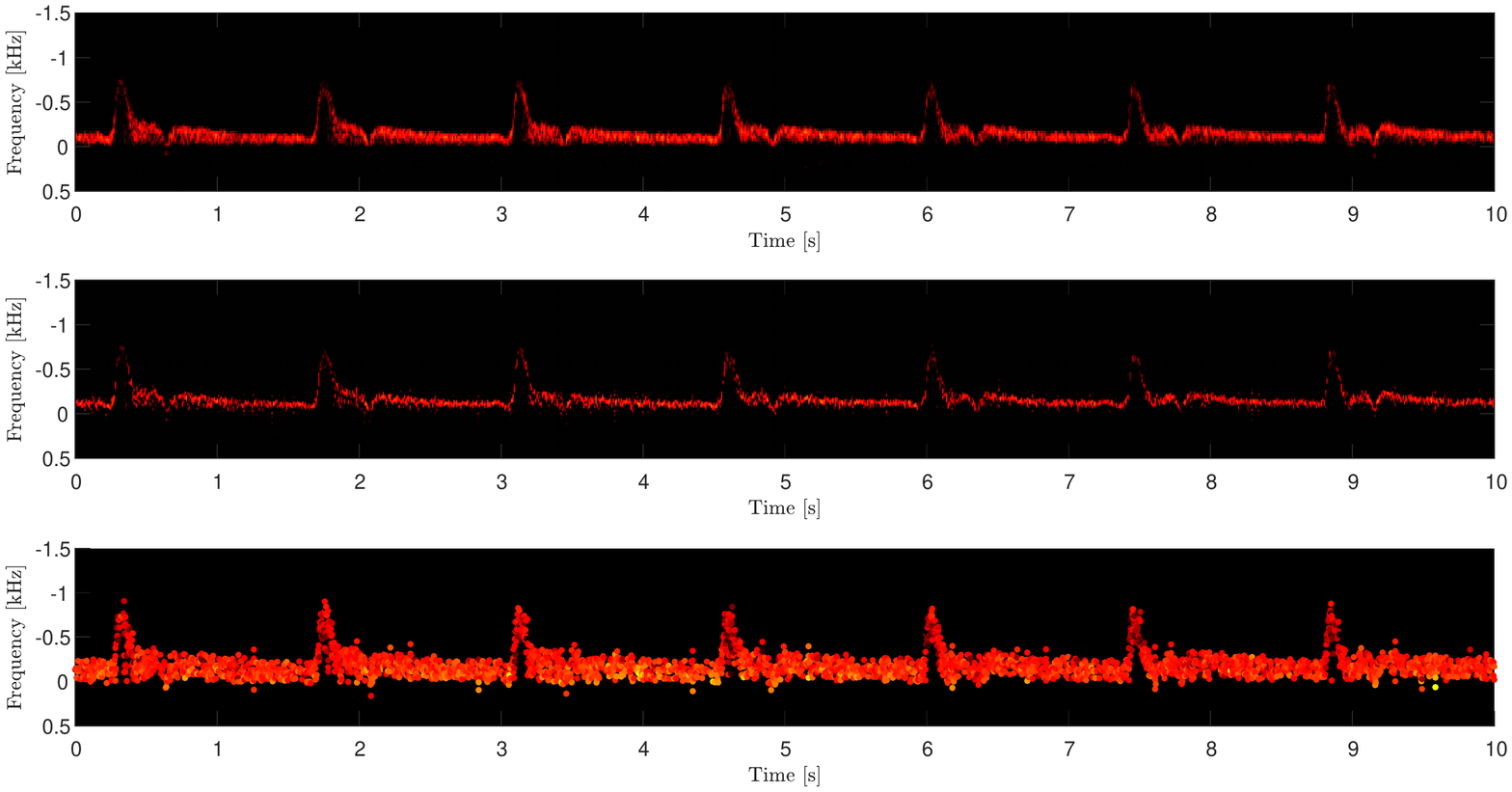}
\caption{{\bf In vivo.} Spectrograms of in vivo data of Carotid artery using 35 out 128 ($\sim 27\%$) according the nested transmission pattern. (Top) Welch's method (middle) NEST (bottom) NESPRIT. All spectrograms are displayed with a dynamic range of 60 dB.}
\label{fig:invivo}
\end{figure*}

\section{Conclusion}
\label{sec:con}
In this paper, we presented a sparse irregular transmit scheme for medical spectral Doppler based on nested arrays. Using this approach,  we showed that in noiseless settings the blood spectrum can be recovered from only $2\sqrt{P}-1$ emissions, where $P$ is the size of the observation window. Two recovery algorithms NEST and NESPRIT, which exploit the proposed transmission pattern, were presented. NEST exhibits low complexity and performs efficient reconstruction of the blood spectrum with enhanced resolution. NESPRIT  theoretically achieves infinite frequency-precision in recovering the blood velocities at the expense of computational load. Moreover, any clutter filter and apodization function can be easily incorporated into NEST and NESPRIT. Both algorithms were evaluated and tested with Field II simulation data of pulsating flow from the femoral artery. NEST and NESPRIT were compared and shown to outperform current state-of-the-art methods by successfully recovering the blood spectrum from only 12\% of the Doppler transmissions. Finally, in vivo results showed the ability of the proposed techniques to yield valid spectrograms using far fewer emissions, proving their potential for clinical use. This paves the way for a duplex mode displaying high resolution blood spectrograms while providing high quality B-mode images at a high frame rate.

\appendices
\section{Proof of Lemma 1}
\label{app:lemma1}
First, it easy to see that the maximal difference in absolute values between elements of $\mathcal{S}_\text{N}$ is $N_2(N_1+1)-1$. Hence, there is no integer $k$ such that $|k| >N_2(N_1+1)-1$ which belongs to $\mathcal{D}$  or $\mathcal{D}_u$.

Given any integer $k$ in the range $-N_2(N_1+1)+1\leq k\leq N_2(N_1+1)-1$,  we have that $k\in\mathcal{D}_u$ if there exists $p_i$ and $p_j$ which satisfy
\begin{equation}
k = p_i-p_j,\quad p_i,p_j\in\mathcal{S}_\text{N}.
\label{eq:kinS}
\end{equation}
Note that if for a specific $k$ there exists such $p_i,p_j\in\mathcal{S}_\text{N}$ , i.e., $k\in\mathcal{D}_u$, then also $-k\in\mathcal{D}_u$ since
\begin{equation}
-k = -(p_i-p_j)=p_j-p_i.
\end{equation} 
Therefore, we focus on proving (\ref{eq:kinS}) only for non-negative integers $k$ in the range $0\leq k\leq N_2(N_1+1)-1$. 

Every integer $k$ in the desired range can be decomposed as
\begin{equation}
k = m(N_1+1)+r,
\label{eq:k}
\end{equation}
where $m$ and $r$ are integers in the ranges $0\leq m\leq N_2-1$ and $0\leq r\leq N_1$, respectively. Denoting $p_i=(m+1)(N_1+1)$ and $p_j=N_1+1-r$, we can rewrite (\ref{eq:k}) as
\begin{align}
\begin{split}
k &= m(N_1+1)+r = \\ &=(m+1)(N_1+1)+r - N_1-1= \\
&= (m+1)(N_1+1) - (N_1+1-r) = \\
&= p_i-p_j.
\end{split}
\end{align}
By definition $p_i\in\mathcal{S}_{\text{N}_2}$. When $r=0$, $p_j\in\mathcal{S}_{\text{N}_2}$; otherwise $p_j\in\mathcal{S}_{\text{N}_1}$. Thus, $p_i,p_j\in\mathcal{S}_\text{N}$ and we conclude that $k\in\mathcal{D}_u$.

\section{Proof of Theorem \ref{theo:mintrans}}
\label{app:theo1}
Denoting $\tilde{N_1}=N_1+1$, we recast problem (\ref{eq:mintrans}) as follows
\begin{align}
\begin{split}
\underset{\substack{\tilde{N_1},N_2\in\mathbb{N}^{+} \\ \tilde{N_1}\geq 2  }}{\min}\qquad &\tilde{N_1}+N_2 -1\\ 
\text{subject to}\qquad &N_2\tilde{N_1}=P.
\end{split}
\label{eq:mintrans2}
\end{align} 
From (\ref{eq:mintrans2}), it is easy to see that $N_2=\frac{P}{\tilde{N_1}}$, where  $\tilde{N_1}$ is a divisor of $P$. Assuming $\tilde{N_1}\leq N_2$, we have 
\begin{equation}
\tilde{N_1}=\underset{\mathcal{D}_1}{\text{argmin}} \quad d+\frac{P}{d},
\end{equation}
where we neglect the constant term $-1$.

Next, we define a function $f:[1,\sqrt{P}\,]\rightarrow\mathcal{R}^+$ over a continuous domain
\begin{equation*}
f(x)=x+\frac{P}{x}.
\end{equation*} 
The function $f(x)$ is continuous and differentiable over the open
interval $(1,\sqrt{P})$ .
Its derivative is given by
\begin{equation*}
\frac{df}{dx}=1-\frac{P}{x^2}< 0,
\end{equation*}
hence, $f(x)$ is monotonically decreasing. Since $\mathcal{D}_1\subset[1,\sqrt{P}\,]$, denoting $\tilde{N_1}=\max(\mathcal{D}_1)$, we have
\begin{equation*}
f(\tilde{N_1})< f(d),\quad d\in\mathcal{D}_1.
\end{equation*}
Therefore, the optimal solution is given by $N_1=\max(\mathcal{D}_1)-1$ and $N_2=\frac{P}{\tilde{N_1}}=\min(\mathcal{D}_2)$ accordingly.
By interchanging the roles of  $\tilde{N_1}$ and $N_2$ we get the solution for $\tilde{N_1}\geq N_2$, given by $N_1=\min(\mathcal{D}_2)-1$ and $N_2=\max(\mathcal{D}_1)$.

\section{Proof of Theorem \ref{theo:kmintrans}}
\label{app:theo3}
First, consider a given K-level nested array with $L$ levels and $\{N_i\}_{i=1}^L$. Notice that if $N_K=1$, then the resulting geometry can be seen as a nested array with $L-1$ levels and $\{\tilde{N}_i\}_{i=1}^{L-1}$ where
\begin{align*}
&\tilde{N}_i=N_i,\quad i=1,...,L-2, \\
&\tilde{N}_{L-1}=N_{L-1}+1.
\end{align*}  
Therefore, we assume that $N_K>1$ to avoid ambiguity.

For simplicity of analysis, we define
\begin{equation*}
Z_i\triangleq
\begin{cases}
N_i+1,\quad i=1,2,..,K-1,\\
N_K,\qquad\, i=K.
\end{cases}
\end{equation*}
An equivalent problem to (\ref{eq:kmintrans}) can be rewritten as
\begin{align}
\begin{split}
\underset{K\in\mathbb{N}^{+}}{\min}\quad\underset{\substack{Z_1,...,Z_K\in\mathbb{N}^{+} \\ Z_1,...,Z_{K}\geq 2}}{\min}\qquad &\sum_{i=1}^{K} Z_i-K+1 \\ 
\text{subject to}\qquad\qquad &\prod_{i=1}^{K}Z_i=P.
\end{split}
\label{eq:kmintransapp}
\end{align}
Following (\ref{eq:kmintransapp}), we wish to prove that $K=\Omega$ and the optimal $\{Z_i\}_{i=1}^\Omega$ are given by the prime factors of $P$ with repetitions according to their multiplicities (up to rotation).


Assume by contradiction that the optimal solution satisfies  $K\neq \Omega$. The fundamental theorem of arithmetic states that every positive integer has a single unique prime factorization \cite{riesel1994prime}, hence, $K\leq\Omega$. Assume $K<\Omega$, then there exists $Z_i$ which is not prime, i.e., $Z_i$ can be decomposed into the multiplication of two smaller integers $Z_{i1}$ and $Z_{i2}$ where $Z_{i1},Z_{i2}\geq2$. This amounts to breaking the $i$th nested level into 2 levels such that now there are $K+1$ levels of nesting. Assuming that $Z_{i1}\leq Z_{i2}$  without loss of generality, we have
\begin{equation}
Z_{i1}+Z_{i2}\leq 2Z_{i2} \leq Z_{i1}Z_{i2}.
\end{equation}
Hence, breaking up the $i$th nesting level into two levels decreases the value of the objective function in contradiction to the optimality of the solution. 

Following the latter, we should go on splitting the nesting levels till all $Z_i$ are prime numbers. This, along with the fact that the prime factorization is unique, implies that the total number of levels of nesting is $K=\Omega$  and the optimal $\{Z_i\}_{i=1}^\Omega$ are the prime factors of $P$.


\ifCLASSOPTIONcaptionsoff
  \newpage
\fi



\bibliographystyle{IEEEtran}
\bibliography{IEEEabrv,refs_jrnl}

\begin{thebibliography}{10}
\providecommand{\url}[1]{#1}
\csname url@samestyle\endcsname
\providecommand{\newblock}{\relax}
\providecommand{\bibinfo}[2]{#2}
\providecommand{\BIBentrySTDinterwordspacing}{\spaceskip=0pt\relax}
\providecommand{\BIBentryALTinterwordstretchfactor}{4}
\providecommand{\BIBentryALTinterwordspacing}{\spaceskip=\fontdimen2\font plus
\BIBentryALTinterwordstretchfactor\fontdimen3\font minus
  \fontdimen4\font\relax}
\providecommand{\BIBforeignlanguage}[2]{{%
\expandafter\ifx\csname l@#1\endcsname\relax
\typeout{** WARNING: IEEEtran.bst: No hyphenation pattern has been}%
\typeout{** loaded for the language `#1'. Using the pattern for}%
\typeout{** the default language instead.}%
\else
\language=\csname l@#1\endcsname
\fi
#2}}
\providecommand{\BIBdecl}{\relax}
\BIBdecl

\bibitem{jensen1996estimation}
J.~A. Jensen, \emph{Estimation of blood velocities using ultrasound: A signal
  processing approach}.\hskip 1em plus 0.5em minus 0.4em\relax Cambridge
  University Press, 1996.

\bibitem{welch1967use}
P.~Welch, ``The use of fast {F}ourier transform for the estimation of power
  spectra: A method based on time averaging over short, modified
  periodograms,'' \emph{IEEE Transactions on Audio and Electroacoustics},
  vol.~15, no.~2, pp. 70--73, 1967.

\bibitem{stoica2005spectral}
P.~Stoica, R.~L. Moses \emph{et~al.}, \emph{Spectral Analysis of
  Signals}.\hskip 1em plus 0.5em minus 0.4em\relax Pearson Prentice Hall Upper
  Saddle River, NJ, 2005, vol. 452.

\bibitem{kristoffersen1988time}
K.~Kristoffersen and B.~A. Angelsen, ``A time-shared ultrasound {D}oppler
  measurement and 2-{D} imaging system,'' \emph{IEEE Transactions on Biomedical
  Engineering}, vol.~35, no.~5, pp. 285--295, 1988.

\bibitem{klebaek1995neural}
H.~Kleb{\ae}k, J.~A. Jensen, and L.~K. Hansen, ``Neural network for sonogram
  gap filling,'' in \emph{International Ultrasonics Symposium}, vol.~2.\hskip
  1em plus 0.5em minus 0.4em\relax IEEE, 1995, pp. 1553--1556.

\bibitem{jensen2006spectral}
J.~A. Jensen, ``Spectral velocity estimation in ultrasound using sparse data
  sets,'' \emph{The Journal of the Acoustical Society of America (ASA)}, vol.
  120, no.~1, pp. 211--220, 2006.

\bibitem{mollenbach2008duplex}
S.~K. Mollenbach and J.~A. Jensen, ``Duplex scanning using sparse data
  sequences,'' in \emph{International Ultrasonics Symposium}.\hskip 1em plus
  0.5em minus 0.4em\relax IEEE, 2008, pp. 5--8.

\bibitem{liu2009periodically}
P.~Liu and D.~Liu, ``Periodically gapped data spectral velocity estimation in
  medical ultrasound using spatial and temporal dimensions,'' in
  \emph{International Conference on Acoustics, Speech and Signal Processing
  (ICASSP)}.\hskip 1em plus 0.5em minus 0.4em\relax IEEE, 2009, pp. 437--440.

\bibitem{larsson2003spectral}
E.~G. Larsson and J.~Li, ``Spectral analysis of periodically gapped data,''
  \emph{IEEE Transactions on Aerospace and Electronic Systems}, vol.~39, no.~3,
  pp. 1089--1097, 2003.

\bibitem{gudmundson2011blood}
E.~Gudmundson, A.~Jakobsson, J.~A. Jensen, and P.~Stoica, ``Blood velocity
  estimation using ultrasound and spectral iterative adaptive approaches,''
  \emph{Signal Processing}, vol.~91, no.~5, pp. 1275--1283, 2011.

\bibitem{gran2009adaptive}
F.~Gran, A.~Jakobsson, and J.~A. Jensen, ``Adaptive spectral {D}oppler
  estimation,'' \emph{IEEE Transactions on Ultrasonics, Ferroelectrics, and
  Frequency Control}, vol.~56, no.~4, 2009.

\bibitem{torp1997clutter}
H.~Torp, ``Clutter rejection filters in color flow imaging: A theoretical
  approach,'' \emph{IEEE Transactions on Ultrasonics, Ferroelectrics, and
  Frequency Control}, vol.~44, no.~2, pp. 417--424, 1997.

\bibitem{eldar2012compressed}
Y.~C. Eldar and G.~Kutyniok, \emph{Compressed sensing: Theory and
  Applications}.\hskip 1em plus 0.5em minus 0.4em\relax Cambridge University
  Press, 2012.

\bibitem{zobly2011compressed}
S.~M. Zobly and Y.~M. Kakah, ``Compressed sensing: Doppler ultrasound signal
  recovery by using non-uniform sampling \& random sampling,'' in \emph{28th
  National Radio Science Conference (NRSC)}.\hskip 1em plus 0.5em minus
  0.4em\relax IEEE, 2011, pp. 1--9.

\bibitem{zobly2013multiple}
S.~M. Zobly and Y.~M. Kadah, ``Multiple measurements vectors compressed sensing
  for {D}oppler ultrasound signal reconstruction,'' in \emph{International
  Conference on Computing, Electrical and Electronics Engineering
  (ICCEEE)}.\hskip 1em plus 0.5em minus 0.4em\relax IEEE, 2013, pp. 319--322.

\bibitem{demanet2007wave}
L.~Demanet and L.~Ying, ``Wave atoms and sparsity of oscillatory patterns,''
  \emph{Applied and Computational Harmonic Analysis}, vol.~23, no.~3, pp.
  368--387, 2007.

\bibitem{richy2013blood}
J.~Richy, D.~Friboulet, A.~Bernard, O.~Bernard, and H.~Liebgott, ``Blood
  velocity estimation using compressive sensing,'' \emph{IEEE Transactions on
  Medical Imaging}, vol.~32, no.~11, pp. 1979--1988, 2013.

\bibitem{richy2011blood}
J.~Richy, H.~Liebgott, R.~Prost, and D.~Friboulet, ``Blood velocity estimation
  using compressed sensing,'' in \emph{International Ultrasonics Symposium
  (IUS)}.\hskip 1em plus 0.5em minus 0.4em\relax IEEE, 2011, pp. 1427--1430.

\bibitem{lorintiu2016compressed}
O.~Lorintiu, H.~Liebgott, and D.~Friboulet, ``Compressed sensing {D}oppler
  ultrasound reconstruction using block sparse {B}ayesian learning,''
  \emph{IEEE Transactions on Medical Imaging}, vol.~35, no.~4, pp. 978--987,
  2016.

\bibitem{zhang2013extension}
Z.~Zhang and B.~D. Rao, ``Extension of {SBL} algorithms for the recovery of
  block sparse signals with intra-block correlation,'' \emph{IEEE Transactions
  on Signal Processing}, vol.~61, no.~8, pp. 2009--2015, 2013.

\bibitem{zhang2012recovery}
------, ``Recovery of block sparse signals using the framework of block sparse
  {B}ayesian learning,'' in \emph{International Conference on Acoustics, Speech
  and Signal Processing (ICASSP)}.\hskip 1em plus 0.5em minus 0.4em\relax IEEE,
  2012, pp. 3345--3348.

\bibitem{pal2010nested}
P.~Pal and P.~P. Vaidyanathan, ``Nested arrays: A novel approach to array
  processing with enhanced degrees of freedom,'' \emph{IEEE Transactions on
  Signal Processing}, vol.~58, no.~8, pp. 4167--4181, 2010.

\bibitem{pal2010novel}
P.~Pal and P.~Vaidyanathan, ``A novel array structure for directions-of-arrival
  estimation with increased degrees of freedom,'' in \emph{International
  Conference on Acoustics Speech and Signal Processing (ICASSP)}.\hskip 1em
  plus 0.5em minus 0.4em\relax IEEE, 2010, pp. 2606--2609.

\bibitem{jensen1996field}
J.~A. Jensen, ``Field: A program for simulating ultrasound systems,'' in
  \emph{10th Nordicbaltic Conference on Biomedical Imaging, vol. 4, supplement
  1, part 1: 351--353}.\hskip 1em plus 0.5em minus 0.4em\relax Citeseer, 1996.

\bibitem{jensen1992calculation}
J.~A. Jensen and N.~B. Svendsen, ``Calculation of pressure fields from
  arbitrarily shaped, apodized, and excited ultrasound transducers,''
  \emph{IEEE Transactions on Ultrasonics, Ferroelectrics, and Frequency Control
  (UFFC)}, vol.~39, no.~2, pp. 262--267, 1992.

\bibitem{van2002optimum}
H.~L. Van~Trees, \emph{Optimum Array Processing. Part IV of Detection,
  Estimation, and Modulation Theory.}\hskip 1em plus 0.5em minus 0.4em\relax
  New York: Wiley Intersci., 2002.

\bibitem{ma2009doa}
W.-K. Ma, T.-H. Hsieh, and C.-Y. Chi, ``Doa estimation of quasi-stationary
  signals via {K}hatri-{R}ao subspace,'' in \emph{International Conference on
  Acoustics, Speech and Signal Processing (ICASSP)}.\hskip 1em plus 0.5em minus
  0.4em\relax IEEE, 2009, pp. 2165--2168.

\bibitem{ma2010doa}
------, ``Doa estimation of quasi-stationary signals with less sensors than
  sources and unknown spatial noise covariance: {A} {K}hatri--{R}ao subspace
  approach,'' \emph{IEEE Transactions on Signal Processing}, vol.~58, no.~4,
  pp. 2168--2180, 2010.

\bibitem{pal2012nested}
P.~Pal and P.~Vaidyanathan, ``Nested arrays in two dimensions, {P}art {I}:
  Geometrical considerations,'' \emph{IEEE Transactions on Signal Processing},
  vol.~60, no.~9, pp. 4694--4705, 2012.

\bibitem{pal2012nested-b}
------, ``Nested arrays in two dimensions, {P}art {II}: Application in two
  dimensional array processing,'' \emph{IEEE Transactions on Signal
  Processing}, vol.~60, no.~9, pp. 4706--4718, 2012.

\bibitem{montaldo2009coherent}
G.~Montaldo, M.~Tanter, J.~Bercoff, N.~Benech, and M.~Fink, ``Coherent
  plane-wave compounding for very high frame rate ultrasonography and transient
  elastography,'' \emph{IEEE Transactions on Ultrasonics, Ferroelectrics, and
  Frequency Control}, vol.~56, no.~3, pp. 489--506, 2009.

\bibitem{pal2014soft}
P.~Pal and P.~Vaidyanathan, ``Soft-thresholding for spectrum sensing with
  coprime samplers,'' in \emph{8th Sensor Array and Multichannel Signal
  Processing Workshop (SAM)}.\hskip 1em plus 0.5em minus 0.4em\relax IEEE,
  2014, pp. 517--520.

\bibitem{donoho1995noising}
D.~L. Donoho, ``De-noising by soft-thresholding,'' \emph{IEEE Transactions on
  Information Theory}, vol.~41, no.~3, pp. 613--627, 1995.

\bibitem{pal2018guarantees}
A.~Koochakzadeh and P.~Pal, ``Non-asymptotic guarantees for correlation-aware
  support detection,'' in \emph{International Conference on Acoustics, Speech
  and Signal Processing (ICASSP)}.\hskip 1em plus 0.5em minus 0.4em\relax IEEE,
  2018.

\bibitem{chi2011sensitivity}
Y.~Chi, L.~L. Scharf, A.~Pezeshki, and A.~R. Calderbank, ``Sensitivity to basis
  mismatch in compressed sensing,'' \emph{IEEE Transactions on Signal
  Processing}, vol.~59, no.~5, pp. 2182--2195, 2011.

\bibitem{pal2014grid}
P.~Pal and P.~Vaidyanathan, ``A grid-less approach to underdetermined direction
  of arrival estimation via low rank matrix denoising,'' \emph{IEEE Signal
  Processing Letters}, vol.~21, no.~6, pp. 737--741, 2014.

\bibitem{pal2014gridless}
------, ``Gridless methods for underdetermined source estimation,'' in
  \emph{48th Asilomar Conference on Signals, Systems and Computers}.\hskip 1em
  plus 0.5em minus 0.4em\relax IEEE, 2014, pp. 111--115.

\bibitem{liu2015remarks}
C.-L. Liu and P.~Vaidyanathan, ``Remarks on the spatial smoothing step in
  coarray {MUSIC},'' \emph{IEEE Signal Processing Letters}, vol.~22, no.~9, pp.
  1438--1442, 2015.

\bibitem{eldar2015sampling}
Y.~C. Eldar, \emph{Sampling Theory: Beyond Bandlimited Systems}.\hskip 1em plus
  0.5em minus 0.4em\relax Cambridge University Press, 2015.

\bibitem{roy1989esprit}
R.~Roy and T.~Kailath, ``Esprit-estimation of signal parameters via rotational
  invariance techniques,'' \emph{IEEE Transactions on Acoustics, Speech, and
  Signal Processing}, vol.~37, no.~7, pp. 984--995, 1989.

\bibitem{pan1999complexity}
V.~Y. Pan and Z.~Q. Chen, ``The complexity of the matrix eigenproblem,'' in
  \emph{Proceedings of the thirty-first annual ACM symposium on Theory of
  computing}.\hskip 1em plus 0.5em minus 0.4em\relax ACM, 1999, pp. 507--516.

\bibitem{xu1994fast}
G.~Xu and T.~Kailath, ``Fast subspace decomposition,'' \emph{IEEE Transactions
  on Signal Processing}, vol.~42, no.~3, pp. 539--551, 1994.

\bibitem{lediju2009sources}
M.~A. Lediju, B.~C. Byram, and G.~E. Trahey, ``Sources and characterization of
  clutter in cardiac b-mode images,'' in \emph{International Ultrasonics
  Symposium (IUS)}.\hskip 1em plus 0.5em minus 0.4em\relax IEEE, 2009, pp.
  1419--1422.

\bibitem{hedrick1995ultrasound}
W.~R. Hedrick, D.~L. Hykes, and D.~E. Starchman, \emph{Ultrasound Pysics and
  Instrumentation: Practice examinations}.\hskip 1em plus 0.5em minus
  0.4em\relax CV Mosby, 1995.

\bibitem{avdal2015effects}
J.~Avdal, L.~Lovstakken, and H.~Torp, ``Effects of reverberations and clutter
  filtering in pulsed {D}oppler using sparse sequences,'' \emph{IEEE
  Transactions on Ultrasonics, Ferroelectrics, and Frequency Control}, vol.~62,
  no.~5, pp. 828--838, 2015.

\bibitem{alfred2010eigen}
C.~Alfred and L.~Lovstakken, ``Eigen-based clutter filter design for ultrasound
  color flow imaging: a review,'' \emph{IEEE Transactions on Ultrasonics,
  Ferroelectrics, and Frequency Control}, vol.~57, no.~5, 2010.

\bibitem{liu2016superconf}
C.-L. Liu and P.~Vaidyanathan, ``Super nested arrays: Sparse arrays with less
  mutual coupling than nested arrays,'' in \emph{International Conference on
  Acoustics, Speech and Signal Processing (ICASSP)}.\hskip 1em plus 0.5em minus
  0.4em\relax IEEE, 2016, pp. 2976--2980.

\bibitem{liu2016super}
------, ``Super nested arrays: Linear sparse arrays with reduced mutual
  coupling—{P}art {I}: Fundamentals,'' \emph{IEEE Transactions on Signal
  Processing}, vol.~64, no.~15, pp. 3997--4012, 2016.

\bibitem{liu2016high}
------, ``Super nested arrays: Linear sparse arrays with reduced mutual
  coupling—{P}art {II}: High-order extensions,'' \emph{IEEE Transactions on
  Signal Processing}, vol.~64, no.~16, pp. 4203--4217, 2016.

\bibitem{superMatlab}
\BIBentryALTinterwordspacing
------. (2016) Super nested program. [Online]. Available:
  \url{http://systems.caltech.edu/dsp/students/clliu/SuperNested/SN.zip}
\BIBentrySTDinterwordspacing

\bibitem{merrill2001introduction}
I.~S. Merrill \emph{et~al.}, ``Introduction to radar systems,'' \emph{Mc
  Grow-Hill}, 2001.

\bibitem{vaidyanathan2013coprime}
P.~Vaidyanathan and P.~Pal, ``Coprime sampling and arrays in one and multiple
  dimensions,'' in \emph{Multiscale Signal Analysis and Modeling}.\hskip 1em
  plus 0.5em minus 0.4em\relax Springer, 2013, pp. 105--137.

\bibitem{vaidyanathan2011sparse}
P.~P. Vaidyanathan and P.~Pal, ``Sparse sensing with co-prime samplers and
  arrays,'' \emph{IEEE Transactions on Signal Processing}, vol.~59, no.~2, pp.
  573--586, 2011.

\bibitem{pal2011coprime}
P.~Pal and P.~P. Vaidyanathan, ``Coprime sampling and the {MUSIC} algorithm,''
  in \emph{Digital Signal Processing Workshop and IEEE Signal Processing
  Education Workshop (DSP/SPE)}, 2011, pp. 289--294.

\bibitem{vaidyanathan2012direct}
P.~Vaidyanathan and P.~Pal, ``Direct-{MUSIC} on sparse arrays,'' in
  \emph{International Conference on Signal Processing and Communications
  (SPCOM)}.\hskip 1em plus 0.5em minus 0.4em\relax IEEE, 2012, pp. 1--5.

\bibitem{pal2012correlation}
P.~Pal and P.~Vaidyanathan, ``Correlation-aware techniques for sparse support
  recovery,'' in \emph{Statistical Signal Processing Workshop (SSP)}.\hskip 1em
  plus 0.5em minus 0.4em\relax IEEE, 2012, pp. 53--56.

\bibitem{tan2014direction}
Z.~Tan, Y.~C. Eldar, and A.~Nehorai, ``Direction of arrival estimation using
  co-prime arrays: A super resolution viewpoint,'' \emph{IEEE Transactions on
  Signal Processing}, vol.~62, no.~21, pp. 5565--5576, 2014.

\bibitem{vaidyanathan2011theory}
P.~Vaidyanathan and P.~Pal, ``Theory of sparse coprime sensing in multiple
  dimensions,'' \emph{IEEE Transactions on Signal Processing}, vol.~59, no.~8,
  pp. 3592--3608, 2011.

\bibitem{womersley1955xxiv}
J.~R. Womersley, ``Oscillatory motion of a viscous liquid in a thin-walled
  elastic tube—{I}: The linear approximation for long waves,'' \emph{The
  London, Edinburgh, and Dublin Philosophical Magazine and Journal of Science},
  vol.~46, no. 373, pp. 199--221, 1955.

\bibitem{riesel1994prime}
H.~Riesel, ``Prime numbers and computer methods for factoring,'' Royal
  Institute of Technology in Sweden: Birkhauser Boston 1994.

\end{thebibliography}

\end{document}